\documentclass[11pt]{article}
\usepackage[utf8]{inputenc}
\usepackage[
	persons={Nima, Vishesh, Fred, Huy, June},
	authors=blocks,
    bibliosources={refs.bib}
]{pomegranate}

\usepackage{enumitem}
\usepackage{appendix}
\usepackage{nicematrix}
\usetikzlibrary{hobby,chains}

\DeclareOperator{\diag}
\DeclareOperator{\TV}
\DeclareOperator{\OP}
\DeclareOperator{\mix}
\DeclareOperator{\Ent}
\DeclareDocumentMathCommand{\inflMat}{}{\Psi^{\text{inf}}}
\DeclareDocumentMathCommand{\corMat}{}{\Psi^{\text{cor}}}
\DeclareDocumentMathCommand{\DKL}{}{\D_{\operatorname{KL}}}

\tikzset{pics/VertexGraph/.style args={#1#2#3#4#5#6}{code={
	\begin{scope}[every node/.style={circle, inner sep=3, line width=1, draw=Black}]
		\tikzset{Style-/.style={fill=White}}
		\tikzset{Style0/.style={fill=Gray}}
		\tikzset{Style1/.style={fill=Orange}}
		\tikzset{Style*/.style={fill=Gray}}
		\node[Style#1] (a) at (-30: 0.5) {};
		\node[Style#2] (b) at (90: 0.5) {};
		\node[Style#3] (c) at (-150: 0.5) {};
		\node[Style#4] (d) at (150: 1) {};
		\node[Style#5] (e) at (-90: 1) {};
		\node[Style#6] (f) at (30: 1) {};
	\end{scope}
	\foreach \u/\v in {a/b, b/c, c/a, d/b, d/c, e/a, e/c, f/a, f/b}
		\draw[line width=1] (\u) -- (\v);
}}}

\title{Entropic Independence I: Modified Log-Sobolev Inequalities for Fractionally Log-Concave Distributions and High-Temperature Ising Models}
\author[1]{Nima Anari}
\author[1]{Vishesh Jain}
\author[2]{Frederic Koehler}
\author[1]{Huy Tuan Pham}
\author[1]{Thuy-Duong Vuong}
\affil[1]{Stanford University, \url{{anari,visheshj,huypham,tdvuong}@stanford.edu}}
\affil[2]{Simons Institute, UC Berkeley, \url{fkoehler@berkeley.edu}}
\date{}

\begin{document}
	\maketitle

	\begin{abstract}
		We introduce a notion called entropic independence that is an entropic analog of spectral notions of high-dimensional expansion. Informally, entropic independence of a background distribution $\mu$ on $k$-sized subsets of a ground set of elements  says that for any (possibly randomly chosen) set $S$, the relative entropy of a single element of $S$ drawn uniformly at random carries at most $O(1/k)$ fraction of the relative entropy of $S$, a constant multiple of its ``share of entropy.'' Entropic independence is the natural analog of the recently established notion of spectral independence, if one replaces variance by entropy. We use entropic independence to derive tight mixing time bounds for natural random walks associated with a number of widely studied distributions, overcoming the lossy nature of spectral analysis of Markov chains on exponential-sized state spaces. In our main technical result, we show a general way of deriving entropy contraction, a.k.a.\ modified log-Sobolev inequalities, for down-up random walks from much simpler spectral notions. We show that spectral independence of a distribution under arbitrary external fields automatically implies entropic independence. Our result can be seen as a new framework to establish entropy contraction from the much simpler variance contraction inequalities.
		
		To derive our results, we relate entropic independence to properties of polynomials: $\mu$ is entropically independent exactly when a transformed version of the generating polynomial of $\mu$ is upper bounded by its linear tangent; this property is implied by concavity of the said transformation, which was shown by prior work to be locally equivalent to spectral independence. Our framework makes no assumptions on marginals of $\mu$ or the degrees of the underlying graphical model when $\mu$ is based on one, and it has the ability to derive tight bounds on mixing time even when it is not nearly-linear. We apply our results to obtain tight modified log-Sobolev inequalities and mixing times for multi-step down-up walks on fractionally log-concave distributions. As our flagship application, we establish the tight mixing time of $O(n\log n)$ for Glauber dynamics on Ising models whose interaction matrix  has eigenspectrum lying within an interval of length smaller than $1$, improving upon the prior quadratic dependence on $n$.
	\end{abstract}
	\clearpage
	
	\section{Introduction}



Let
$\mu:\binom{[n]}{k} \to \R_{\geq 0}$ be a non-negative density on the $k$-subsets of the ground set $[n] = \set{1,\dots,n}$. Such a density naturally defines a distribution on the $k$-subsets of $[n]$ given by
\[\P{S} \propto \mu(S).\]
We study a family of local Markov chains that can be used to approximately sample from such a distribution.

\begin{definition}[Down-up random walks]\label{def:local-walk}
	For a density $\mu:\binom{[n]}{k}\to\R_{\geq 0}$, and an integer $\l\leq k$, we define the $k\leftrightarrow\l$ down-up random walk as the sequence of random sets $S_0, S_1,\dots$ generated by the following algorithm:
	\begin{Algorithm*}
		\For{$t=0,1,\dots$}{
			Select $T_t$ uniformly at random from subsets of size $\l$ of $S_t$.\;
			Select $S_{t+1}$ with probability $\propto \mu(S_{t+1})$ from supersets of size $k$ of $T_t$.\;
		}
	\end{Algorithm*}
\end{definition}

This random walk is time-reversible and always has $\mu$ as its stationary distribution \cite[see, e.g.,][]{ALO20}. The special case of $\l=k-1$ has received the most attention, especially in the literature on high-dimensional expanders \cite[see, e.g.,][]{LLP17,KO18,DK17,KM16}. Recent works have established the utility of down-up random walks in capturing and analyzing widely studied Markov chains such as Glauber dynamics on graphical models or basis-exchange random walks on matroids \cite{ALOV19,CGM19,AL20,ALO20,CLV20,alimohammadi2021fractionally,CGSV21,FGYZ21,JPV21,liu2021coupling,blanca2021mixing}.

Each step of the down-up random walk can be efficiently implemented with oracle access to $\mu$ as long as $k-\l=O(1)$. This is because  the number of supersets of $T_t$ is at most $n^{k-\l}=\poly(n)$, so we can enumerate over all of them in polynomial time. Even though the $k\leftrightarrow \l$ random walk is interesting algorithmically only when $\l=k-O(1)$, the entire range of down-up random walks is useful as an analysis tool. In fact, analyzing down-up walks where $\l=1$, and concluding mixing time of $k\leftrightarrow k-O(1)$ random walks, is the key technique behind most of the high-dimensional-expanders-based breakthroughs in Markov chain analysis \cite[see, e.g.,][]{DK17,KO18,AL20}. 

In this work, we introduce the notion of entropic independence as a tool to establish tight bounds on the mixing time of the down-up random walks via lower bounding the modified log-Sobolev constant \cite{bobkov2006modified} (see \cref{def:MLSI}). Entropic independence is an entropy-based analog of spectral notions of high-dimensional expansion such as local spectral expansion \cite{KM16,DK17,KO18} and spectral independence \cite{ALO20}. The motivation behind considering entropy-based notions is that Markov chain mixing time analysis via spectral techniques is often lossy, by polynomially large factors on exponential-sized state spaces. On the other hand, entropy-based analysis of Markov chains can often yield tight mixing time bounds \cite{bobkov2006modified}.

Our work introduces a novel technique to get \emph{optimal} mixing time bounds using the rapidly growing literature on high-dimensional-expanders-based Markov chain analysis. Related prior works in this area fall into two categories:
\paragraph{Prior work on matroids and log-concave polynomials.} \Textcite{CGM19} established nearly-linear mixing time of $\tilde O(k)$\footnote{The notation $\tilde O(\cdot)$ suppresses logarithmic factors in both $k$ and $n$. To keep the exposition simple in the introduction, we assume the random walks are started from a starting point with large enough of a probability mass; there always exists a point $S$ where $\P_{\mu}{S}\geq 1/\binom{n}{k}$ and we assume our starting point approximately has this mass in the reported mixing times.} for the $k\leftrightarrow k-1$ down-up random walk whenever $\mu$ has a log-concave generating polynomial and the walk is started from a good point. This improved upon the earlier bound of $\tilde O(k^2)$ established via spectral analysis by \textcite{ALOV19}. Matroids are extremely good high-dimensional expanders \cite{ALOV19}; unfortunately the techniques of \textcite{CGM19} appear to be limited to just matroids and matroid-related distributions, since they crucially use a specific threshold of high-dimensional expansion that can only be achieved by these distributions.

\paragraph{Prior work on bounded-degree graphical models.} \Textcite{CLV20,blanca2021mixing} showed that under certain assumptions, spectral independence, a condition weaker than the extremely good high-dimensional expansion of matroids \cite{ALO20}, implies nearly-linear mixing time of $\tilde O(k)$ for down-up walks. This yielded breakthrough tight mixing time bounds for a wide range of distributions originating mostly from statistical physics. However, these results need some key assumptions, most of which do not appear to be inherently necessary, and seem to be just needed for the proof to work. The main assumption behind these works is that $\mu$ captures a graphical model on a bounded degree graph; that is $\mu$ is the joint distribution of a collection of $k$ random variables arranged as nodes of a $O(1)$-degree graph, and that any two regions of the graph are conditionally independent of each other, conditioning on a set of nodes separating them. The bounded degree assumption together with the conditional independence assumption allows for $\mu$ to shatter into small independent pieces of size $\simeq \log k$ after conditioning on linearly many variables. A further assumption is that each random variable has domain of size at most $O(1)$ and that each element in the domain has marginal $\geq \Omega(1)$ according to $\mu$ and conditionings of $\mu$. None of the assumptions, being graphical, being constant-degree, or having large marginals, appear to be inherently necessary, but are crucial to the proof.

\paragraph{This work.} In this work we show that good high-dimensional expansion \emph{under external fields} (see \cref{def:external-field}) automatically implies entropy contraction inequalities, and in particular modified log-Sobolev inequalities for down-up random walks. Unlike mentioned prior works, we do not require extreme high-dimensional expansion, we do not assume conditional independence (being graphical), we do not require a lower bound on the marginals of $\mu$, or even in the case of graphical models for the degrees of the graph to be $O(1)$; instead we require good high-dimensional expansion under \emph{external fields}. This assumption is strictly weaker than the extreme high-dimensional expansion assumption of \cite{CGM19}, so our results give a proper generalization of the main result of \cite{CGM19}. A distinguishing feature of our techniques is the ability to derive the optimal mixing time when it is not nearly-linear (see \cref{sec:app-flc}). Another feature of our techniques is the ability to handle lopsided distributions, where $\min\set{\mu(S)\given \mu(S)>0}$ is extremely small. Lopsided distributions provably cannot have a large (non-modified) log-Sobolev constant \cite{bobkov2006modified}, but they can possibly have a large \emph{modified} log-Sobolev constant. This distinction between modified and non-modified log-Sobolev inequalities is a barrier for the prior techniques of \cite{CLV20,blanca2021mixing}, since they yield roughly the same log-Sobolev and modified log-Sobolev constants in the settings where they can be applied.

Informally, we call a background measure $\mu: \binom{[n]}{k} \to \R_{\geq 0}$ entropically independent if for any (possibly randomly chosen) set $S$, the relative entropy of an element of $S$ drawn uniformly at random carries at most $O(1/k)$ fraction of the relative entropy of $S$, a constant multiple of its ``share of entropy.'' More precisely, entropic independence can be defined as entropy contraction of the $D_{k \to 1}$ operator, where $D_{k \to \l}$ is the first part of the $k\leftrightarrow \l$ random walk, i.e., it operates on a set $S\in \binom{[n]}{k}$ by uniformly sampling a size-$\l$ subset of $S$. Note that $D_{k \to \l}$ sends a distribution $\mu$ over $\binom{[n]}{k}$ to the distribution $\mu D_{k \to \l}$ over $\binom{[n]}{\l}.$  
\begin{definition}[Entropic independence]
    \label{def:entropic-independence}
A probability distribution $\mu$ on $\binom{[n]}{k}$ is said to be $(1/\alpha)$-entropically independent, for $\alpha \in (0,1]$, if for all probability distributions $\nu$ on $\binom{[n]}{k}$,
\[ \DKL{\nu D_{k\to 1} \river \mu D_{k\to 1}}\leq \frac{1}{\alpha k}\DKL{\nu \river \mu}.  \]
\end{definition}
Entropic independence is a natural analog for spectral independence, another recently established notion by \textcite{ALO20}, if one replaces variance by entropy. For the special case where $\mu$ is defined via a graphical model, notions like entropic independence have been studied in prior works (although mostly as an interesting corollary, and not as the main tool to establish mixing times). See \cite[e.g.,][]{blanca2021mixing} for the notion of approximate subadditivity of entropy.

\paragraph{Spectral vs.\ entropic independence.}
In \cite{ALO20}, spectral independence is defined as an upper bound on the spectral norm of the pairwise correlation matrix of $\mu$ (\cref{def:corr}), or equivalently, an upper bound on the second largest eigenvalue of the simple (non-lazy) random walk on the
$1$-skeleton of $\mu$, when viewing $\mu$ as a weighted high-dimensional expander \cite{KO18, DK17}.\footnote{In the original definition \cite{DK17, KO18,ALO20}, such a requirement is imposed for both $\mu$ and all links of $\mu$, where the link of $\mu$ w.r.t.\ a set $T$ is the distribution of $S-T$ given that $S$ is sampled from $\mu$ conditioned on $T\subseteq S$. For the sake of clarity, and to avoid unnecessary assumptions on uniformity over links, we take the lone term ``spectral independence'' to only refer to the link of $T=\emptyset$, with the understanding that one usually requires spectral independence for all links; similarly, to derive mixing time bounds for down-up walks, we require entropic independence for all links.} The simple random walk on the
$1$-skeleton of $\mu$ samples from $\mu D_{k \to 1}$ by transitioning from $\set*{i}$ to $\set*{j}$ with probability proportional to $\mu D_{k \to 2} (\set*{i, j}).$ An upper bound on the second largest eigenvalue of this random walk is equivalent (up to a simple linear transformation) to an upper bound on the second largest eigenvalue of the $k \leftrightarrow 1$-down-up walk to sample from $\mu.$ One walk is simply a lazier version of the other. Standard results about the relationship between second largest eigenvalue and variance contraction then imply that variance contraction of $D_{2 \to 1}$ with respect to $\mu D_{k \to 2}$ is equivalent to variance contraction of $D_{k \to 1}$ with respect to $\mu.$ However, such an equivalence does not hold when we replace variance with entropy: entropy contraction of $D_{2 \to 1}$ with respect to $\mu D_{k \to 2}$ is a stronger assumption than entropic independence. See \cref{example:no-2-to-1} in \cref{sec:bad-example} for why natural distributions might not have good $D_{2\to 1}$ entropy contraction. This introduces an inherent difficulty in establishing entropic independence. While spectral independence is about an $n\times n$ matrix, or the expansion of a simple graph on $n$ nodes (whose edges are given by $\mu D_{k\to 2}$), there is no such compact object determining entropic independence; one has to look at all the $k$-sized sets and the full distribution $\mu$.


We connect entropic independence to the geometry of the generating polynomial of the distribution $\mu$. The multivariate generating polynomial $g_{\mu} \in \R[z_1,\dots, z_n]$ associated to $\mu: \binom{[n]}{k} \to \R_{\geq 0}$ is given by
\[g_{\mu}(z_1,\dots, z_n) := \sum_{S}\mu(S)\prod_{i\in S}z_i.\]

\begin{definition}[External field]\label{def:external-field}
For a distribution $\mu$ on $\binom{[n]}{k}$ and $\lambda = (\lambda_1,\dots, \lambda_n) \in \R^{n}_{>0}$, the notation $\lambda \ast \mu$ denotes the distribution $\mu$ tilted by external field $\lambda$, which is a distribution on $\binom{[n]}{k}$ given by 
\[\P_{\lambda \ast \mu}{S} \propto \mu(S)\cdot \prod_{i \in S}\lambda_i.\]
\end{definition}
Note that for any $(z_1,\dots, z_n) \in \R^{n}_{\geq 0}$, 
\[g_{\lambda \ast \mu}(z_1,\dots, z_n) \propto g_{\mu}(\lambda_1 z_1,\dots, \lambda_n z_n).\]

	In \cref{thm:entropic-independence}, we show that a distribution $\mu$ is entropically independent exactly when a transformed version of the generating polynomial of $\mu$ can be upper bounded by its linear tangent, a property implied by concavity of the said transformation. We further show that this concavity is equivalent to fractional log-concavity \cite{alimohammadi2021fractionally}, which is in turn equivalent to spectral independence under arbitrary external fields.

We recall the definition of $\alpha$-fractional-log-concavity \cite{alimohammadi2021fractionally}.
	
	For $\alpha \in (0,1]$, a distribution $\mu$ on $\binom{[n]}{k}$ is said to be $\alpha$-fractionally log-concave (abbreviated as $\alpha$-FLC) if $\log g_{\mu}(z_1^{\alpha},\dots, z_{n}^{\alpha})$ is concave, viewed as a function on $\R^{n}_{\geq 0}$. We note that $1$-FLC is equivalent to complete/strong log-concavity \cite{ALOV19,BH19}.
		\begin{theorem} \label{thm:entropic-independence}
		Let $\mu$ be a probability distribution on $\binom{[n]}{k}$ and let $p = (p_1,\dots, p_n) := \mu D_{k \to 1}$. Then, for $\alpha \in (0,1]$,
		\[\mu \text{ is }(1/\alpha)\text{-entropically independent}\iff \forall (z_1,\dots, z_n) \in \R^{n}_{\geq 0},\text{ } g_{\mu}(z_1^\alpha,\dots, z_n^\alpha)^{1/k\alpha} \leq \sum_{i=1}^{n}p_i z_{i}.\]
		In particular, if $\mu$ is $\alpha$-fractionally log-concave, then $\mu$ is $(1/\alpha)$-entropically independent. Moreover,
			\[\lambda \ast \mu \text{ is }(1/\alpha)\text{-entropically independent }\forall \lambda = (\lambda_1,\dots, \lambda_n) \in \R^{n}_{>0}\iff \mu \text{ is }\alpha\text{-fractionally log-concave}.\]
	\end{theorem}
	
	\Cref{fig:pictorial} shows a comparison of the notions in terms of the transformed generating polynomial \[g_\mu(z_1^\alpha,\dots,z_n^\alpha)^{1/k\alpha}.\] Spectral independence is equivalent to \emph{local} concavity of this function around the point $(1,\dots,1)$ \cite{alimohammadi2021fractionally}. \Cref{thm:entropic-independence} shows that entropic independence is equivalent to the linear tangent at $(1,\dots,1)$ upper bounding the function. Fractional log-concavity is equivalent to concavity of the entire function at all points in the positive orthant; roughly speaking, external fields allow us to replace the $(1,\dots,1)$ point by any other point in the positive orthant.
	
	\begin{figure}
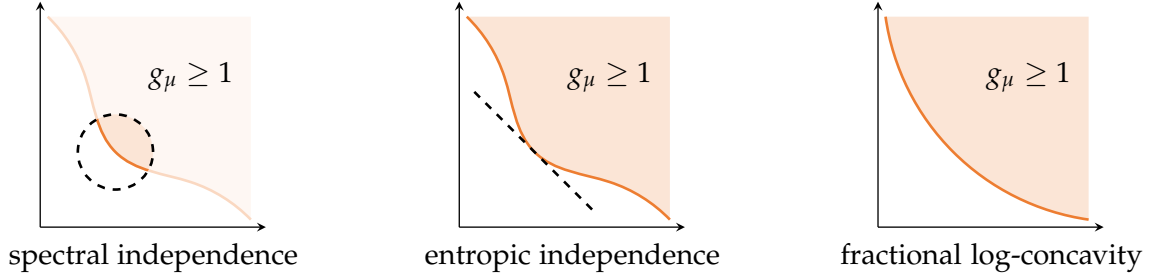

		\begin{Columns}<3>
		\Column
		\Tikz*[use Hobby shortcut]{
			\draw[-stealth, line width=0.5, Black] (0, 0) -- (3, 0);
			\draw[-stealth, line width=0.5, Black] (0, 0) -- (0, 3);
			\path[fill=LightOrange!30!White] (2.8, 0.1) .. (2, 0.6) .. (1, 1) .. (0.6, 2) .. (0.1, 2.8) -- (2.8, 2.8) -- cycle;
			\draw[Orange!30!White, line width=1] (2.8, 0.1) .. (2, 0.6) .. (1, 1) .. (0.6, 2) .. (0.1, 2.8);
			\begin{scope}
				\clip (1, 1) circle (0.5);
				\path[fill=LightOrange] (2.8, 0.1) .. (2, 0.6) .. (1, 1) .. (0.6, 2) .. (0.1, 2.8) -- (2.8, 2.8) -- cycle;
				\draw[Orange, line width=1] (2.8, 0.1) .. (2, 0.6) .. (1, 1) .. (0.6, 2) .. (0.1, 2.8);
			\end{scope}
			\draw[dashed, line width=1] (1, 1) circle (0.5);
			\node at (2, 2) {$g_\mu\geq 1$};
		}
		\Column
		\Tikz*[use Hobby shortcut]{
			\draw[-stealth, line width=0.5, Black] (0, 0) -- (3, 0);
			\draw[-stealth, line width=0.5, Black] (0, 0) -- (0, 3);
			\path[fill=LightOrange] (2.8, 0.1) .. (2, 0.6) .. (1, 1) .. (0.6, 2) .. (0.1, 2.8) -- (2.8, 2.8) -- cycle;
			\draw[Orange, line width=1] (2.8, 0.1) .. (2, 0.6) .. (1, 1) .. (0.6, 2) .. (0.1, 2.8);
			\draw[dashed, line width=1, Black] (0.2, 1.8) -- (1.8, 0.2);
			\node at (2, 2) {$g_\mu\geq 1$};
		}
		\Column
		\Tikz*[use Hobby shortcut]{
			\draw[-stealth, line width=0.5, Black] (0, 0) -- (3, 0);
			\draw[-stealth, line width=0.5, Black] (0, 0) -- (0, 3);
			\path[fill=LightOrange] (2.8, 0.1) .. (1, 1) .. (0.1, 2.8) -- (2.8, 2.8) -- cycle;
			\draw[Orange, line width=1] (2.8, 0.1) .. (1, 1) .. (0.1, 2.8);
			\node at (2, 2) {$g_\mu\geq 1$};
		}
		\end{Columns}
		\begin{Columns}<3>
			\Column
				\Center{spectral independence}
			\Column
				\Center{entropic independence}
			\Column
				\Center{fractional log-concavity}			
		\end{Columns}
		\caption{\label{fig:pictorial}
			The axes are $z_1$, $z_2$, etc. The highlighted area is a level set of the transformed generating polynomial, e.g., where $g_\mu(z_1^\alpha,\dots,z_n^\alpha) \geq 1$. For a degree $k\alpha$-homogeneous function like $g_\mu(z_1^\alpha,\dots,z_n^\alpha)$, log-concavity, concavity of $g_\mu^{1/k\alpha}$, and quasi-concavity, that is convexity of the level sets, are all equivalent (folklore, see \cref{lem:lc-equiv}).  Notions get stronger from left to right. Spectral independence means that around the point $(1, \dots, 1)$, the level set is locally convex. Entropic independence means that the entire level set is globally above its tangent at $(1,\dots,1)$. Fractional log-concavity means that the level set is globally convex.
		}
	\end{figure}
	
	\paragraph{The importance of the external fields assumption.} One might a priori wish for entropy contraction or optimal mixing times from just spectral independence with no extra assumptions. In fact, \textcite{liu2021coupling} conjectured that if $\mu$ is an $O(1)$-spectrally independent distribution, the down-up walk for sampling from $\mu$ has modified log-Sobolev constant $\Omega(1/k)$. We refute this conjecture. Any $\alpha$-fractionally-log-concave distribution is also $O(1/\alpha)$-spectrally independent \cite[Remark~68]{alimohammadi2021fractionally}. However, there are examples of $\Omega(1)$-fractionally-log-concave distributions for which the down-up walk is not even irreducible --- see \cref{remark:lowerbound on downup walk}. Even ignoring the ergodicity issue of the walk, it is well-known that the spectral gap of a Markov chain does not \emph{automatically} translate into a modified log-Sobolev inequality \cite[see, e.g.,][]{bobkov2006modified}. A classical example is the random walk on a constant-degree expander graph. If $G$ is an $n$-node constant degree expander, then the random walk on $G$ has $\Omega(1)$ spectral gap, but its mixing time is $\simeq \log(n)$; a constant factor entropy contraction of this random walk would imply a mixing time of $\simeq \log\log(n)$ which is clearly wrong. One can view (the lazy version of) this random walk as a special case of down-up random walks. Let $\mu:\binom{[n]}{2}\to \R_{\geq 0}$ be the uniform distribution on the edges of the expander. Then the $2\leftrightarrow 1$ down-up walk is the same as the lazy random walk on $G$ itself. This shows that even for down-up random walks, one cannot obtain entropic independence from spectral independence with no extra assumptions.

		
	As mentioned above, for general spectrally independent distributions, the $k\leftrightarrow k-1$ down-up walk might not even be ergodic. We can fix the lack of ergodicity in the $k \leftrightarrow k - 1$ down-up walk by considering the more general $k \leftrightarrow \ell$ down-up walk for smaller values of $\ell$.
    For this more general walk, we establish a modified log-Sobolev inequality for all fractionally log-concave distributions via entropic independence.
	More precisely, in \cref{thm:alpha frac LC MLSI}, we show that for $\alpha$-fractionally-log-concave $\mu: \binom{[n]}{k} \to \R_{\geq 0}$, the $k \leftrightarrow (k-\lceil 1/\alpha\rceil)$ down-up walk \cite{AL20,alimohammadi2021fractionally}, has modified log-Sobolev constant $\Omega(k^{-1/\alpha}).$ The dependence on $\alpha$ is optimal --- again, see \cref{remark:lowerbound on downup walk}. 
	
	\cref{thm:alpha frac LC MLSI} is a natural generalization of a recent result of \textcite[Theorem~1]{CGM19}, which shows that the modified log-Sobolev constant of the $k \leftrightarrow (k-1)$-down-up walk for sampling from a log-concave (i.e., $\alpha=1$) distribution  $\mu:\binom{[n]}{k} \to \R_{\geq 0}$ is at least $\Omega(\frac{1}{k})$. We remark that our proof for \cref{thm:alpha frac LC MLSI} is shorter and less ``mysterious`` than that of \cite[Theorem~1]{CGM19}.
	\begin{theorem} \label{thm:alpha frac LC MLSI}
	Suppose $\mu: \binom{[n]}{k} \to \R_{\geq 0}$ is $\alpha$-fractionally log-concave, or more generally $(1/\alpha)$-entropically independent for all links. 
	Let $\l \leq k-\ceil{1/\alpha}.$ 
	 For all probability distributions $\nu$ on $\binom{[n]}{k}$,
\[ \DKL{\nu D_{k\to \l} \river \mu D_{k\to \l}}\leq (1- \kappa)\DKL{\nu \river \mu}.  \]
 Consequently, the $k \leftrightarrow \l$ down-up walk w.r.t.\ $\mu$ has modified log-Sobolev constant 
	$\geq \Omega(\kappa)$ where $\Omega$ hides an absolute constant and \[\kappa =\frac{  (k+1 - \l -  1/\alpha )^{1/\alpha -\lceil 1/\alpha \rceil } \prod_{i=0}^{\lceil 1/\alpha \rceil -1}(k-\l-i) }{(k+1)^{1/\alpha}}.\]
	For $1/\alpha\in \Z$ we can obtain a simpler looking alternative bound of
	\[ \kappa=\left.{\binom{k-\l}{1/\alpha}}\middle /{\binom{k}{1/\alpha}}\right. . \]
	
	\end{theorem}
	
	\begin{remark}\label{remark:lowerbound on downup walk}
    	\cref{thm:alpha frac LC MLSI} is tight up to a constant. For $\alpha=1$, it is easy to come up with a log-concave distribution $\mu$ for which the one-step down-up walk has modified log-Sobolev constant $\Theta(\frac{1}{k})$, e.g., $\mu$ uniformly supported on $\set*{[k], [k-1] \cup \set*{k+1} }.$ For $\alpha = 1/r$ with $r \in \N_{\geq 1},$ we can consider the following example. Let $\mu$ be the uniform distribution on $\binom{[n]}{k'}$, and we will define a corresponding $r$-fold distribution $\mu^{(r)}$ supported on 
    	\[ \set*{S \times [r] \given S \in \supp(\mu)}\subseteq \binom{[n]\times [r]}{k'r}. \]
    	The probability mass function on its support is given by
    	\[ \mu^{(r)}(S \times [r]) = \mu(S), \]
    	where $ S \times [r]$ is the set $\set*{(a, i) \given a\in S, i\in [r]}.$ Since the 
    	generating polynomial of $\mu^{(r)}$ is determined by that of $\mu$, which is 1-FLC, we can check that this distribution is indeed $\alpha = 1/r$ fractionally log-concave. 
    	
    	Let $k = k'r$. Clearly, the $k \leftrightarrow (k-r')$-down-up walk w.r.t $\mu^{(r)}$ with $r' < r$ is not even irreducible, because the down-up walk always stays at the same place. On the other hand, by comparing the  $k \leftrightarrow k-r$-down-up walk w.r.t $\mu^{(r)}$ with the $k' \leftrightarrow (k'-1)$-down-up walk w.r.t $\mu$, we can show that the former has modified log-Sobolev constant $\Theta\parens*{1/\binom{k'r}{r}} = \Theta\parens*{\frac{r!}{k^{1/\alpha}}} $.
	\end{remark}
	
	    
	
	

    \subsection{Application: fractionally log-concave polynomials}\label{sec:app-flc}
    
   	The most straightforward application of our techniques is to establish modified log-Sobolev inequalities and tight mixing times for distributions that are spectrally independent under arbitrary external fields, a.k.a.\ fractionally log-concave distributions \cite{alimohammadi2021fractionally}. In \cref{sec:app-ising}, we demonstrate an application beyond this setting, where one has to combine our techniques with others to establish tight mixing time bounds.
   	
   	\Cref{thm:alpha frac LC MLSI} is a generalization of the main results of \textcite{CGM19}, so as a special case we recover the tight mixing time and MLSI constants established previously for all distributions with a log-concave generating polynomial \cite[see][for examples]{ALOV19}. Here we highlight two important examples of fractionally log-concave distributions that go beyond simple log-concavity. We obtain, for the first time, tight mixing time bounds and MLSI constants for these examples; for further examples of fractionally log-concave distributions see \cite{alimohammadi2021fractionally}. We emphasize that in both examples, the tight mixing time is near-quadratic and notably not near-linear. As a result, none of the previous high-dimensional-expanders-based frameworks could obtain the tight mixing time in these examples.
   	
	\begin{definition}[Monomer-Dimer Systems] Suppose that a graph $G=(V, E)$ is given together with node weights $z:V\to \R_{\geq 0}$ and edge weights $w:E\to \R_{\geq 0}$. Then the monomer-dimer system is the distribution on matchings $M$ of the graph where
	\[ \P{\text{matching }M}\propto \prod_{e\in M} w(e)\cdot \prod_{v\text{ not matched by }M}z(v). \]	
		For a matching $M$, the edges in the matching are called dimers, and the nodes outside of the matching are called monomers. For an arbitrary monomer-dimer system, we can define a measure $\mu:\binom{V\times \set{0,1}}{\card{V}}\to \R_{\geq 0}$ capturing the distribution of monomers
		\[ \mu(S)=\begin{cases}
			0& \exists v\in V: \card{S\cap \set{(v, 0), (v, 1)}}\neq 1,\\
			\sum\set*{\text{weight}(M)\given \text{monomers of }M=S\cap (V\times \set{1}) }& \text{otherwise.}\\
		\end{cases} \]
	\end{definition}
	
	\begin{figure}
	    \centering
	    \begin{Columns}[Bottom]
	    \Column
	    \Tikz*{
			\begin{scope}[start chain, node distance=0.2]
				\foreach \t in {110000,101000,111100}
					\node[on chain, circle, draw=Gray, fill=LightGray, line width=1] (U\t) {
						\Tikz{\pic[scale=0.5]{VertexGraph/.expand once=\t};}
					};
			\end{scope}
			\begin{scope}[start chain, node distance=0.2]
				\node[on chain] at (-0.2, -2) {};
				\foreach \t in {1--000,1-1-00} 
					\node[on chain, circle, draw=Gray, fill=LightGray, line width=1] (D\t) {
						\Tikz{\pic[scale=0.5]{VertexGraph/.expand once=\t};}
					};
			\end{scope}
			
			\foreach \t/\s in {110000/1--000, 101000/1--000, 101000/1-1-00, 111100/1-1-00}
				\draw[draw=Gray, line width=1] (U\t) -- (D\s);
	    }
	    \Column
	    \[ L=\begin{bNiceMatrix}
	        0 & 1 & 0 & 0 & \Cdots & 0 & 0\\
	        -1 & 0 & 0 & 0 & \Cdots & 0 & 0\\
	        0 & 0 & 0 & 1 & \Cdots & 0 & 0 \\
	        0 & 0 & -1 & 0 & \Cdots & 0 & 0\\
	        \Vdots & \Vdots & \Vdots & \Vdots & \Ddots & \Vdots & \Vdots \\
	        0 & 0 & 0 & 0 & \Cdots & 0 & 1\\
	        0 & 0 & 0 & 0 & \Cdots & -1 & 0\\ 
	    \end{bNiceMatrix} \]
	    \end{Columns}
	    \begin{Columns}[Top]
	    \Column
	    \caption{2-site Glauber dynamics on monomers. A configuration consists of an assignment of binary labels indicating monomer/non-monomer to the vertices. In each step, two uniformly randomly picked vertices have their labels resampled conditioned on all other labels, with probabilities dictated by the monomer distribution.}
	    \label{fig:monomer-dimer}
	    \Column
	    \caption{Example of a matrix $L$ with quadratic mixing time for down-up walks. Note that $L+L^\intercal =0 \succeq 0$. $L$ is block-diagonal with $2\times 2$ blocks. The nonsymmetric determinantal point process selects independently and uniformly  at random between including/excluding both elements $\set{2i, 2i+1}$ for each $i$.}
	    \label{fig:nonsymmetric-dpp}
	    \end{Columns}
	\end{figure}
	
	\Textcite{alimohammadi2021fractionally} showed that the distribution of monomers $\mu$ is $\alpha=1/2$-fractionally-log-concave and proceeded to use the $\card{V}\leftrightarrow \card{V}-2$ down-up random walk, a.k.a.\ the 2-site Glauber dynamics (see \cref{fig:monomer-dimer}), to show polynomial-time sampling algorithms for monomer-dimer systems on planar graphs. \Textcite{alimohammadi2021fractionally} showed a spectral gap of $\Omega(1/\card{V}^2)$, and as a result a mixing time of $\tilde O(\card{V}^3)$ assuming a good starting point, for this random walk. As a direct corollary of \cref{thm:alpha frac LC MLSI}, we obtain the following result:
	\begin{corollary}\label{cor:monomer-mixing}
	    The 2-site Glauber dynamics on the monomer distribution of a $k$ vertex graph has MLSI constant $\Omega(1/k^2)$. As a result it mixes in time $\tilde O(k^2)$, assuming the walk is started from a configuration with probability mass $\geq 1/2^{\poly(k)}$.
	\end{corollary}
	Note that the requirement on the starting point is fairly weak (and can even be weakened further to $1/2^{2^{\poly\log(k)}}$). In particular, if we find the matching $M$ that has the maximum monomer-dimer weight using a maximum weighted matching algorithm, and start the random walk from monomers of $M$, that automatically satisfies the initial condition.
	
	It is not hard to see that \cref{cor:monomer-mixing} is tight. Consider the case where the graph $G$ itself is a perfect matching. Then, each step of the 2-site Glauber dynamics has only a $1/k$ chance of picking two endpoints of the same edge; if that does not happen, the resampling of the two vertices does not change anything and a turn is ``wasted.'' It is also not hard to see that roughly speaking all of the $\Omega(k)$ pairs of endpoints of edges need to be resampled once before mixing. So clearly, the mixing time of this chain is $\tilde \Omega(k^2)$.

	Our next application involves nonsymmetric determinantal point processes, a generalization of determinantal point processes which have found many uses in machine learning, recommender systems, and randomized linear algebra \cite[see, e.g., ][]{gartrell2019learning}.
	\begin{definition}[Nonsymmetric Determinantal Point Process]
	    Suppose that $L \in \R^{k\times k}$ is such that $L+L^\intercal\succeq 0$. The distribution on subsets $T$ of $\set{1,\dots,k}$, giving a weight of $\det(L_{T, T})$ to each set, is called the (nonsymmetric) determinantal point process. We view this distribution as $\mu:\binom{[k]\times\set{0,1}}{k}\to \R_{\geq 0}$, where
	    \[ \mu(S)=\begin{cases}
	        0& \exists i: \card{S\cap\set{(i, 0), (i, 1)}}\neq 1,\\
	        \det(L_{T, T})\text{ where }T=\set{i\given (i,1)\in S}& \text{otherwise}.\\
	    \end{cases} \]
	\end{definition}
	\Textcite{alimohammadi2021fractionally} showed that the above distribution is also $\alpha=1/2$-fractionally log-concave, and as a result established a spectral gap of $\Omega(1/k^2)$ and a mixing time of $\tilde O(k^3)$ with a good choice of the starting point, for the $k\leftrightarrow k-2$ down-up random walk.
	
	\begin{corollary}
	    The $k\leftrightarrow k-2$ down-up random walk on nonsymmetric determinantal point processes has MLSI constant $\Omega(1/k^2)$. As a result, it mixes in time $\tilde O(k^2)$, assuming the walk is started from a set with probability mass $\geq 1/2^{\poly(k)}$.
	\end{corollary}
	Again, finding a starting point satisfying the initial condition is not difficult; for example, it can be achieved by a simple local search \cite{anari2021sampling}. Once again, it is not difficult to see that this quadratic mixing time is tight. The matrix $L$ realizing a mixing time of $\tilde \Omega(k^2)$ can be seen in \cref{fig:nonsymmetric-dpp}. The distribution defined by this particular $L$ in \cref{fig:nonsymmetric-dpp} is the same as the one described by the tight example of monomer distributions; we leave this as an exercise to the reader.
	
	\begin{remark}
	    Both examples mentioned above stem from the so-called Hurwitz-stable polynomials \cite[see ][]{alimohammadi2021fractionally}. \Cref{thm:alpha frac LC MLSI} can be applied more generally to get tight mixing time bounds for arbitrary fractionally log-concave polynomials, and as a special case, the so-called sector-stable polynomials \cite{alimohammadi2021fractionally}. However, there currently seems to be a loss in the analysis of \textcite{alimohammadi2021fractionally} when going from sector-stability to fractional log-concavity for general sector-stable polynomials (this loss is avoided in the case of Hurwitz-stability though). As a result, while we get \emph{improved} mixing time bounds in other examples of fractionally log-concave polynomials, the bounds do not always seem to be tight. If the analysis of \cite{alimohammadi2021fractionally} is improved in future works, and one can establish that homogeneous $\alpha$-sector-stable polynomials are indeed $\alpha$-fractionally log-concave, then combined with our results we would get a tight mixing time bound for all sector-stable polynomials.
	\end{remark}
	
	\subsection{Application: high-temperature Ising models}\label{sec:app-ising}
	
	An Ising model is a probability distribution on the discrete hypercube $\{\pm 1\}^{n}$ given by
	\[\mu_{J,h}(x) \propto \exp\left(\frac{1}{2}\langle x, Jx \rangle + \langle h, x \rangle\right),\]
	where $J \in \R^{n\times n}$ is a symmetric matrix (known as the interaction matrix) and $h \in \R^{n}$ is known as the external field. These models are of fundamental importance in statistical physics and other areas; see e.g. \cite{mezard1987spin,lauritzen1996graphical,mezard2009information,talagrand2010mean}.
	
	The Glauber dynamics (or Gibbs sampler) is a simple and very popular discrete time Markov chain to sample from $\mu_{J,h}$. Its transitions may be described as follows: at each step, given the current state $\sigma$, a coordinate $i$ is chosen uniformly at random from $[n]$ and its value is resampled from the conditional distribution $\mu_{J,h}(\cdot \mid \sigma_{-i})$; it can also be viewed as a down-up walk on a homogenized version of $\mu$, see \cref{def:hom}.
	The study of the Glauber dynamics for the Ising model is a classical topic with numerous results: see  e.g.~\cite{martinelli1999lectures,levin2017markov} for an introduction. Rapid mixing of the Glauber chain is interconnected with other structural properties of the Gibbs measure. In particular, it is well-known that the Dobrushin uniqueness condition -- $\|J\|_{\infty \to \infty} < 1$ (equivalently, $\max_{i \in [n]}\sum_{j}|J_{ij}| < 1$) -- implies that the Glauber dynamics mixes in time $O(n\log{n})$. 
	
	Although Dobrushin's condition is a tight condition for rapid mixing in certain cases (e.g.~the Curie-Weiss model, which is the Ising model on the complete graph), there are many interesting cases where Dobrushin's condition is extremely restrictive. An important example is the famous Sherrington-Kirkpatrick (SK) model from spin glass theory \cite{mezard1987spin,talagrand2010mean}, where $J_{ij} \sim N(0,\beta^2/n)$ is a GOE random matrix. In this case, Dobrushin's condition implies rapid mixing only if the inverse temperature $\beta$ is $O(1/\sqrt{n})$, even though it is widely believed that the Glauber dynamics mixes rapidly for all $\beta < c$ for an absolute constant $c > 0$, even conjecturally with $c = 1$ \cite{simons}.
	
	Polynomial time mixing in the SK model was recently established with $c = 1/4$ in recent work of \textcite{eldan2020spectral} 
	as a consequence of the following general result: if the interaction matrix $J$ is positive semidefinite, then the spectral gap of the Glauber dynamics for the Ising model is at least $\frac{1 - \norm{J}_\OP}{n}$, where $\norm{J}_\OP$ denotes the $\l^2 \to \l^2$ operator norm, which, since $J$ is symmetric, coincides with the largest eigenvalue of $J$. The interaction matrix $J$ can always be assumed to be positive semidefinite without loss of generality, because adding a multiple of the identity matrix to $J$ does not change the measure $\mu_{J,h}$.
	By the well-known relationship between the spectral gap and the mixing time, the result of \cite{eldan2020spectral} implies that the mixing time of the Glauber dynamics is $O\parens*{\frac{n}{1-\norm{J}_\OP}\cdot (n + \norm{h}_1)}$, i.e. quadratic in $n$. However, it seemed plausible that the true mixing time was faster than this.
	
	Here, as the main application of our theory, we establish an optimal $O(n \log n)$ bound on the mixing time and also show that the modified log-Sobolev inequality holds:
	\begin{theorem} \label{thm:main}
	    Let $\mu_{J,h}$ denote an Ising measure on $\{\pm 1\}^{n}$ with $0 \preceq J \preceq \norm{J}_\OP I$ and let $P$ denote the transition matrix of the the Glauber dynamics. Then,
	    \begin{enumerate}[label = (\alph*)]
	        \item \label{part:mlsi} The modified log-Sobolev constant $\rho_0(P)$ of $P$ satisfies $\rho_0(P) \geq \frac{1 - \norm{J}_\OP}{n }$.
	        \item \label{part:mixing} The Glauber dynamics mixes in time $O\parens*{ \frac{n \log n}{1 - \norm{J}_\OP}}$.
	    \end{enumerate}
	\end{theorem}
	The MLSI is used in the proof of mixing, but it also has many other useful direct consequences, such as concentration of measure and reverse hypercontractive estimates \cite{van2014probability,mossel2013reverse}.
	
	\cref{thm:main} can be applied to a number of other models of interest. For example, a $d$-regular version of the diluted SK model can be formed
	by taking a random $d$-regular graph and assigning edge weights i.i.d. from $Uni(\pm \beta)$. Combining the above result with a version of Friedman's theorem (see \cite{eldan2020spectral} for details and references) proves the optimal $O(n \log(n))$ mixing time of Glauber dynamics for all $\beta < \frac{1}{4\sqrt{d - 1}}$, whereas Dobrushin's uniqueness condition, or the more precise tree uniqueness criterion (see e.g. \cite{sinclair2014approximation}), holds only when $\beta = O(1/d)$. The reason for the discrepancy is that for these spin glass models, the uniqueness threshold is not the relevant phase transition on the infinite $d$-ary tree. Instead, the properly analogous phase transition concerns the tree with spin-glass boundary conditions \cite{chayes1986mean,pemantle2010critical} or the purity of the limiting Gibbs measure with free boundary conditions \cite{bleher1995purity,ioffe1996extremality,evans2000broadcasting}, otherwise known as the \emph{reconstruction threshold}, which is well past the uniqueness threshold.

	
	\paragraph{Comparison to previous work.} Although sampling from Ising models is a classical topic, only recently in the breakthrough work of \textcite{bauerschmidt2019very} was a polynomial time sampling result established for the SK model with constant $\beta$. Under the same condition $0 \preceq J \prec I$ as above, they showed how to draw a sample from the Ising model by sampling from a related log-concave distribution in $\mathbb{R}^n$ and applying an additional rounding step. They also proved a a version of the log-Sobolev inequality, but their version only implies $e^{O(\sqrt{n})}$ time mixing bounds for the Glauber dynamics in the SK model --- see discussion in \cite{eldan2020spectral}. 
	
	Later, in the work \cite{eldan2020spectral} it was proved that the Glauber dynamics indeed mix in polynomial time. Their result established a reduction for proving functional inequalities to the case where $J$ is rank one, and the $O(n^2)$ mixing time guarantee was established using the Poincar\'e inequality of \cite{hayes2006simple,wu2006poincare}. However, there was no analogous way to establish the MLSI for the class of rank-one Ising models based on existing results. For example, directly applying a state of the art result such as \cite{marton2019logarithmic} gives an MLSI with constant $e^{-\Omega(\sqrt{n})}$. The issue is that in these models, like the SK model itself, the conditional marginals can be very tiny and existing methods are unable to handle this efficiently. In contrast, our approach based on entropic independence requires no assumption on boundedness of the conditional marginals and enables us to prove the MLSI. 

	\subsection{Techniques}
	First, we discuss the idea behind our main result, \cref{thm:entropic-independence}, which in particular establishes entropic
	independence of the probability distribution $\mu$ given fractional log-concavity of the generating polynomial $g_{\mu}$. Recall that $(1/\alpha)$-entropic independence
	holds for $\mu$ if the inequality
	\[ \DKL{\nu D_{k\to 1} \river \mu D_{k\to 1}}\leq \frac{1}{\alpha k}\DKL{\nu \river \mu}  \]
	is true for all probability measures $\nu$. The connection with the generating polynomial appears when we fix the left hand side, or more precisely fix the marginal $q = (q_1,\ldots,q_n) := \nu D_{k \to 1}$, and ask for the \emph{worst-case} choice of $\nu$ given this constraint: the measure $\nu$ minimizing the rhs. This is a minimum relative entropy problem, so
	based on convex duality (\cref{lem:KL-dual}, cf. \cite{singh2014entropy}), we get a formula in terms of ``dual'' variables $\log z_1,\ldots,\log z_n$ corresponding to the constraints; concretely, we have
	\[ \inf\set*{\DKL{\nu \river \mu} \given \nu D_{k\to 1}=q } = -\log\parens*{\inf_{z_1,\dots, z_n>0} \frac{g_\mu(z_1,\dots,z_n)}{z_1^{kq_1}\cdots z_n^{kq_n}} }. \]
	The factor of $k$ in the exponent of each $z_i$ appears because the down operator $D_{k \to 1}$ has a $1/k$ chance of picking any particular element of its input set $S$. Note that the rhs can now be lower bounded by making any choice of the variables $z_1,\ldots,z_n > 0$. By choosing $z$ related to $q$ and $p := \mu D_{k \to 1}$, we can then show that the right hand side can be lower bounded in terms of $\DKL{q \river p} = \DKL{\nu D_{k\to 1} \river \mu D_{k\to 1}}$ if an appropriate convexity inequality holds, and in particular if we have $\alpha$-fractional log-concavity; we leave the details to \cref{sec:entropic-independence}.
	
	Once entropic independence is established, the modified log-Sobolev inequality (MLSI) of \cref{thm:alpha frac LC MLSI} follows by a version of the ``local-to-global'' argument.
	Similar to \cite[Theorem~46]{alimohammadi2021fractionally} and similar arguments in \cite{CLV20,HL20}, we write $\DKL{\nu \river \mu}$ and $ \DKL{\nu D_{k\to \l} \river \mu D_{k\to \l}}$ as a telescoping sum of terms of the form $\DKL{\nu D_{k \to i}\river \mu D_{k \to i}} - \DKL{\nu D_{k \to (i-1)}\river \mu D_{k \to (i-1)}},$  use the entropic independence of $\mu$ and its conditional distributions to derive inequalities involving these terms, and then chain these inequalities together to derive the desired entropy contraction. The main difference between our proof and the existing local-to-global framework \cite{CGM19,AL20,HL20,alimohammadi2021fractionally} is that we use $k\to 1$ entropy contraction, i.e. entropic independence, as the main building block whereas they use $2 \to 1$ entropy contraction. This difference is crucial: nontrivial $2 \to 1$ entropy contraction does not hold for fractionally log-concave distributions. In particular, if $\mu$ is the uniform distribution over $\binom{[n]}{k}$ and $\mu^{(2)}$ is the corresponding 2-fold distribution, then $\mu^{(2)}$ is $1/2$-FLC but it does not satisfy any nontrivial $2 \to 1$ entropy contraction guarantee independent of $n$: see \cref{example:no-2-to-1} in the Appendix for full details. 

	The proof of the MLSI for the Ising model, \cref{thm:main}\ref{part:mlsi}, requires some additional ingredients, including appropriate generalization of the previous results to \emph{non-uniform} versions of fractional log-concavity and entropic independence. 
	In order to prove \cref{thm:main}\ref{part:mlsi}, we first apply the needle decomposition argument in \cite{eldan2020spectral} to reduce to the rank-1 case, i.e. when $J = uu^\intercal$ for some $u\in \R^n$ (see \cref{sec: reduce to rank 1}). To show the modified log-Sobolev inequality for the Ising model $\mu$ with $J = u u^\intercal $, $\norm{u}_2^2 < 1$, we show that $\mu$ satisfies a {non-uniform} version of fractional log-concavity (\cref{prop:frac LC}), then show that this generalized version of fractional log-concavity still implies a (non-uniform version of) $ D_{n\to 1}$ entropy contraction (\cref{thm:entropic-independence,thm:nonuniform FLC to entropy contraction}), and finally use induction on $n$ to prove that $D_{n\to 1}$ entropy contraction implies the modified log-Sobolev inequality, or more precisely, $D_{n \to (n-1)}$ entropy contraction (\cref{prop:rank 1 MLSI}). The proof of non-uniform fractional log-concavity employs the ``coupling to spectral independence'' framework in \cite{blanca2021mixing,liu2021coupling} to prove local log-concavity at the point $(1,\dots, 1).$ This is enough, because of the fact that log-concavity of $\mu$ at any point $(z_1,\dots, z_n)\in \R_{\geq 0}^n$ is equivalent to log-concavity at $(1,\dots, 1)$ for an Ising model $\mu'$ with the same interaction matrix and a shifted external field. 
	
	Once the MLSI is proven, we immediately get fast mixing of the Glauber dynamics as long as the external field $h$ is not too large. It turns out that this dependence on the size of the external field can be entirely eliminated, because the Glauber dynamics quickly reaches a kind of ``warm start''; to prove this,  we show that the Ising model satisfies an ``exchange inequality'' (\cref{prop:exchange}), and then use the framework in \Textcite{logConcaveIV} to deduce fast mixing without any dependence on the external field.
	
	\subsection{Organization}
	In \cref{sec:prelims}, we collect some preliminary notions and results. In \cref{sec:entropic-independence}, we prove our main result, \cref{thm:entropic-independence}. \cref{sec:MLSI-Ising} proves the desired MLSI for rank-1 Ising models with $\|u\|_{2} < 1$; together with the needle decomposition of \cite{eldan2020spectral} (\cref{thm:needle-decomposition}, \cref{cor:needle}), this immediately implies \cref{thm:main}\ref{part:mlsi}. Finally, in \cref{sec:exchange}, we show that the Ising model satisfies the exchange property of \cite{logConcaveIV}, thereby completing the proof of \cref{thm:main}\ref{part:mixing}. 
	
	\Cref{sec:appendix} contains the proofs of \cref{thm:alpha frac LC MLSI} (which follows by combining \cref{thm:entropic-independence} with the local-to-global framework, as presented in \cite{alimohammadi2021fractionally}) and of \cref{thm:nonuniform FLC to entropy contraction}, which is essentially identical to the proof of \cref{thm:entropic-independence}. 
	
	\begin{subsection}{Acknowledgements}
		Nima Anari and Thuy-Duong Vuong are supported by NSF CAREER award CCF-2045354, a Sloan Research Fellowship, and a Google Faculty Research Award.	Frederic Koehler was supported in part by E.\ Mossel's Vannevar Bush Fellowship ONR-N00014-20-1-2826. Huy Tuan Pham is supported by a Two Sigma Fellowship. 
	\end{subsection}

	\section{Preliminaries}
\label{sec:prelims}
\subsection{Markov chains and functional inequalities}

Let $\mu$ and $\nu$ be probability measures on a finite set $\Omega$. The Kullback-Liebler divergence (or relative entropy) between $\nu$ and $\mu$ is given by
\[\DKL{\nu \river \mu} = \sum_{x \in \Omega}\nu(x)\log\left(\frac{\nu(x)}{\mu(x)}\right),\]
with the convention that this is $\infty$ if $\nu$ is not absolutely continuous with respect to $\mu$. By Jensen's inequality, $\DKL{\nu \river \mu} \geq 0$ for any probability measures $\mu, \nu$. The total variation distance between $\mu$ and $\nu$ is given by
\[d_\TV(\mu, \nu) = \frac{1}{2}\sum_{x \in \Omega}|\mu(x) - \nu(x)|.\]

A Markov chain on $\Omega$ is specified by a row-stochastic non-negative transition matrix $P \in \R^{\Omega \times \Omega}$. We refer the reader to \cite{levin2017markov} for a detailed introduction to the analysis of Markov chains. As is common, we will view probability distributions on $\Omega$ as row vectors. Recall that a transition matrix $P$ is said to be reversible with respect to a distribution $\mu$ if for all $x,y \in \Omega$, $\mu(x)P(x,y) = \mu(y)P(y,x)$. In this case, it follows immediately that $\mu$ is a stationary distribution for $P$ i.e.~$\mu P = \mu$. If $P$ is further assumed to be ergodic, then $\mu$ is its unique stationary distribution, and for any probability distribution $\nu$ on $\Omega$, $d_\TV(\nu P^{t}, \mu) \to 0$ as $t \to \infty$. The goal of this paper is to investigate the rate of this convergence. 

\begin{definition}
Let $P$ be an ergodic Markov chain on a finite state space $\Omega$ and let $\mu$ denote its (unique) stationary distribution. For any probability distribution $\nu$ on $\Omega$ and $\epsilon \in (0,1)$, we define
\[t_\mix(P, \nu, \epsilon) = \min\{t\geq 0 \mid d_\TV(\nu P^{t}, \mu)\leq \epsilon\},\]
and
\[t_\mix(P,\epsilon) = \max\set*{\min\set{t\geq 0 \mid d_\TV(\1_{x}P^t, \mu) \leq \epsilon}\given x\in \Omega},\]
where $\1_{x}$ is the point mass supported at $x$. 
\end{definition}

We will drop $P$ and $\nu$ if they are clear from context. Moreover, if we do not specify $\epsilon$, then it is set to $1/4$. This is because the growth of $t_{\operatorname{mix}}(P,\epsilon)$ is at most logarithmic in $1/\epsilon$ (cf.~\cite{levin2017markov}). 

The modified log-Sobolev constant of a Markov chain, defined next, provides control on its mixing time.

\begin{definition}
\label{def:MLSI}
Let $P$ denote the transition matrix of an ergodic, reversible Markov chain on $\Omega$ with stationary distribution $\mu$.

\begin{itemize}
    \item The Dirichlet form of $P$ is defined for $f,g:\Omega \to \mathbb{R}$ by
    \[\mathcal{E}_P(f,g) = \dotprod{f, (I-P)g}_{\mu} = \dotprod{(I-P)f, g}_{\mu}.\]
    \item The modified log-Sobolev constant of $P$ is defined to be
\[\rho_0(P) = \inf\set*{\frac{\mathcal{E}_P(f, \log f)}{2\cdot \Ent_{\mu}[f]} \given f \colon \Omega \to \mathbb{R}_{\geq 0}, \Ent_{\mu}[f]\neq 0},\]
where 
\[\Ent_{\mu}[f] = \E_{\mu}{f\log f} - \E_{\mu}{f}\log\E_{\mu}{f}.\]
Note that, by rescaling, the infimum may be restricted to functions $f\colon \Omega \to \mathbb{R}_{\geq 0}$ satisfying $\text{Ent}_{\mu}[f]\neq 0$ and $\E_{\mu}{f} = 1$.
\end{itemize}

The relationship between the modified log-Sobolev constant and mixing times is captured by the following well-known lemma. 

\begin{lemma}[{\cite[see, e.g.,][]{bobkov2006modified}}]
Let $P$ denote the transition matrix of an ergodic, reversible Markov chain on $\Omega$ with stationary distribution $\mu$ and let $\rho_0(P)$ denote its modified log-Sobolev constant. Then, for any probability distribution $\nu$ on $\Omega$ and for any $\epsilon \in (0,1)$
\[t_\mix(P, \nu, \epsilon) \leq \ceil*{\rho_0(P)^{-1}\cdot \parens*{\log\log\max\set*{\parens*{\frac{\nu(x)}{\mu(x)}}\given x\in \Omega} + \log\parens*{\frac{1}{2\epsilon^2}}} }.\]
In particular,
\[t_\mix(P,\epsilon) \leq  \ceil*{ \rho_0(P)^{-1}\cdot \parens*{\log\log\parens*{\frac{1}{\min\set{\mu(x)\given x\in \Omega}}} + \log\parens*{\frac{1}{2\epsilon^2}}} }.\]
\end{lemma}

\end{definition}
The next lemma, which shows that contraction of relative entropy under $P$ implies a modified log-Sobolev inequality, is standard. We include the proof for the reader's convenience. 
\begin{lemma}
\label{lem:entropy-contraction-implies-mlsi}
Let $\mu$ be a probability measure on the finite set $\Omega$. Let $P$ denote the transition matrix of an ergodic, reversible Markov chain on $\Omega$ with stationary distribution $\mu$. Suppose there exists some $\alpha \in (0,1]$ such that for all probability measures $\nu$ on $\Omega$ which are absolutely continuous with respect to $\mu$, we have
\[\DKL{\nu P \river \mu P} \leq (1-\alpha)\DKL{\nu \river \mu}.\]
Then, 
\[\rho_0(P) \geq 2\cdot \alpha.\]
\end{lemma}
\begin{proof} 

Let $f\colon \Omega \to \mathbb{R}_{\geq 0}$ with $\Ent_{\mu}[f] \neq 0$ and $\E_{\mu}{f} = 1$. Then, $\nu = f\mu$ is also a probability measure on $\Omega$ with
\[\Ent_{\mu}[f] = \E_{\mu}{f\log f} = \DKL{\nu \river \mu}.\]
Since $\nu$ is absolutely continuous with respect to $\mu$, we have by assumption that
\[\DKL{\nu P \river \mu} =\DKL{\nu P \river \mu P} \leq (1-\alpha)\DKL{\nu \river \mu} = \DKL{\nu \river \mu}-\alpha\Ent_{\mu}[f].\]
Therefore,
\begin{align*}
\alpha \Ent_{\mu}[f]
&\leq \DKL{\nu \river \mu} - \DKL{\nu P \river \mu}\\
&= \sum_{x \in \Omega}\nu(x)\log \parens*{\frac{\nu(x)}{\mu(x)}} - \sum_{x\in \Omega}(\nu P)(x)\log\parens*{\frac{(\nu P)(x)}{\mu(x)}}\\
&= \sum_{x\in \Omega}(\nu(I-P))(x)\log\left(\frac{\nu(x)}{\mu(x)}\right) - \DKL{\nu P \river \nu}\\
&\leq \sum_{x\in \Omega}(\nu (I-P))(x)\log f(x)\\
&= \dotprod{(I-P)f, \log f}_{\mu} = \mathcal{E}_{P}(f, \log f). \qedhere 
\end{align*}
\end{proof}

\subsection{Down-up random walks}
Let $0 \leq k \leq n$ and let $\mu \colon \binom{[n]}{k} \to \R_{\geq 0}$ be a non-negative function on the size-$k$ subsets of $[n]$. Note that $\mu$ is naturally associated to a probability distribution on $\binom{[n]}{k}$. We will find it useful to view the Ising model with $n$ spins as a distribution over the size-$n$ subsets of $[2n]$. For a set $\Omega = \{i_1,\dots, i_n\}$, we define the set $\bar{\Omega} = \{\bar{i_1},\dots, \bar{i_n}\}$, which is disjoint from $\Omega$, and each of whose elements is naturally paired with an element of $\Omega$. 

\begin{definition}\label{def:hom}
For $\sigma \in \set*{\pm 1}^{\Omega},$ let $\sigma^{\hom} \in \binom{\Omega \cup \bar{\Omega} }{n}$ be the set $\set*{i \in \Omega \given \sigma_i = 1} \cup\set*{\bar{i} \in \bar{\Omega} \given \sigma_i = -1}.$ For a distribution $\mu$ over $\set*{\pm 1}^{\Omega}$ with $\abs{\Omega} = n$, let the {homogenization} of $\mu$, denoted by $\mu^{\hom},$ be the distribution supported on $\set*{\sigma^{\hom} \given \sigma \in \set*{\pm 1}^{\Omega}} $ defined by $\mu^{\hom}(\sigma^{\hom}) \propto \mu(\sigma). $ 
\end{definition}

We will also find it useful to interpret the Glauber dynamics as the $n \leftrightarrow (n-1)$ down-up walk on $\binom{\Omega \cup \bar{\Omega} }{n}$. Recall that the down-up walk is given by the composition of two row-stochastic operators, known as the down and up operators.
\begin{definition}[Down operator]
	For a ground set $\Omega$, and  $|\Omega| \geq k\geq \l$,  define the down operator $D_{k\to \l}\in \R^{\binom{\Omega}{k}\times \binom{\Omega}{\l}}$ as
	\[ 
		D_{k\to \l}(S, T)=\begin{cases}
			\frac{1}{\binom{k}{\l}}&\text{ if }T\subseteq S,\\
			0&\text{ otherwise}.\\
		\end{cases}
	\]
\end{definition}
Note that $D_{k\to \l}D_{\l\to m}=D_{k\to m}$. 

\begin{definition}[Up operator]
	For a ground set $\Omega$, $|\Omega|\geq k\geq \l$, and density $\mu:\binom{\Omega}{k}\to \R_{\geq 0}$, define the up operator $U_{\l \to k}\in \R^{\binom{\Omega}{\l}\times \binom{\Omega}{k}}$ as
	\[ 
		U_{\l\to k}(T, S)=\begin{cases}
			\frac{\mu(S)}{\sum_{S'\supseteq T}\mu(S')}&\text{ if }T\subseteq S,\\
			0&\text{ otherwise}.\\
		\end{cases}
	\]
\end{definition}
If we define $\mu_k=\mu$ and more generally let $\mu_\l$ be $\mu_k D_{k\to \l}$, then the down and up operators satisfy
\[ \mu_k(S)D_{k\to \l}(S, T)=\mu_\l(T)U_{\l \to k}(T, S). \]
This property ensures that the composition of the down and up operators have the appropriate $\mu$ as a stationary distribution, are reversible, and have nonnegative real eigenvalues.
\begin{proposition}[{\cite[see, e.g.,][]{KO18,AL20,ALO20}}]
	The operators $D_{k\to \l}U_{\l\to k}$ and $U_{\l\to k}D_{k\to \l}$ both define Markov chains that are time-reversible and have nonnegative eigenvalues. Moreover $\mu_k$ and $\mu_\l$ are respectively their stationary distributions.
\end{proposition}

\begin{definition}[Down-up walk]
	For a ground set $\Omega$, $|\Omega|\geq k\geq \l$, and density $\mu:\binom{\Omega}{k}\to \R_{\geq 0}$, the $k\leftrightarrow \l$ down-up walk is defined by the row-stochastic matrix $U_{\l \to k}D_{k\to \l}$.
\end{definition}



\subsection{Polynomials}
\begin{definition}
The multivariate generating polynomial $g_{\mu} \in \R[z_1,\dots, z_n]$ associated to a density $\mu\colon \binom{[n]}{k} \to \R_{\geq 0}$ is given by
\[g_{\mu}(z_1,\dots, z_n) := \sum_{S}\mu(S)\prod_{i\in S}z_i = \sum_{S}\mu(S)z^{S},\]
\end{definition}
Here we have used the standard notation that for $S \subseteq [n]$, $z^S = \prod_{i\in S}z_i$.

In \cite{alimohammadi2021fractionally}, the notion of fractional log-concavity of the multivariate generating polynomial was developed. We will need a slight generalization of this notion.  

\begin{definition}[Non-uniform fractional log-concavity]
Consider a homogeneous distribution $\mu: \binom{[n]}{k} \to \R_{\geq 0}$ and let  $g_{\mu}(z_1, \dots, z_n)$ be its multivariate generating polynomial. For $\vec{\alpha} = (\alpha_1, \dots, \alpha_n) \in [0,1]^{n}$, we say that $\mu$ is $\vec{\alpha}$-fractionally log-concave (abbreviated as $\vec{a}$-FLC) if 
$\log g_{\mu}(z_1^{\alpha_1}, \dots, z_n^{\alpha_n})$ is concave, viewed as a function over $\R_{\geq 0}^{n}$. 
\end{definition}

\begin{remark}
If the distribution $\mu$ on $\binom{[n]}{k}$ is $\vec{\alpha}$-FLC, then the same is true for the conditional distributions $\mu_T$ for all $T \subseteq \binom{[n]}{\leq k}$. Here, $\mu_{T}$ is the distribution on $\binom{[n]\setminus T}{k-|T|}$ with $\mu_T(S)\propto \mu(T \cup S)$. This is because 
\[g_{\mu_T} \propto \lim_{\lambda \to \infty}\lambda^{-|T|}\cdot g_{\mu}(\{\lambda z_i\}_{i\in T}, \{z_i\}_{i\notin T})\]
and operations of scaling the variables or the polynomial and taking limits preserve $\vec{\alpha}$-FLC.  
\end{remark}

We will make use of the following characterization of log-concavity for homogeneous functions. Recall that a function $f:\R_{\geq 0}^n\to\R_{\geq 0}$ is said to be $d$-homogeneous if $f(cx)=c^df(x)$ for all $c>0$. 

\begin{lemma}[Folklore]
\label{lem:lc-equiv}
Let $\mathcal{C} \subseteq \R_{\geq 0}^{n}$ denote a convex cone.  
		For a $d$-homogeneous function $f:\mathcal{C}\to\R_{\geq 0}$ the following are all equivalent:
		\begin{enumerate}
			\item $f$ is quasi-concave.
			\item $f$ is log-concave.
			\item $f$ is $d$-th-root-concave, i.e., $f^{1/d}$ is concave.
		\end{enumerate}  
\end{lemma}
\begin{proof}
    The implications ($3$)$\implies$($2$)$\implies$($1$) are immediate (and hold even without homogeneity). For the implication ($1$)$\implies$($3$), for any $x,y \in \mathcal{C}$ and $\lambda \in [0,1]$, we have by quasi-concavity and homogeneity that
    \begin{align*}
        &\frac{f(\lambda x + (1-\lambda)y)}{(\lambda f(x)^{1/d} + (1-\lambda)f(y)^{1/d})^d}  \\
        &= f\parens*{\frac{\lambda f(x)^{1/d}}{\lambda f(x)^{1/d} + (1-\lambda)f(y)^{1/d}}\cdot \frac{x}{f(x)^{1/d}} + \frac{(1-\lambda) f(y)^{1/d}}{\lambda f(x)^{1/d} + (1-\lambda)f(y)^{1/d}}\cdot \frac{y}{f(y)^{1/d}}}  \\
        &\geq \min\set*{f\parens*{\frac{x}{f(x)^{1/d}}}, f\parens*{\frac{y}{f(y)^{1/d}}}} = 1. 
    \end{align*}
    Rearranging finishes the proof. 
\end{proof}

We will also need the following characterization of the solution of the minimum relative entropy problem with prescribed marginals which can be obtained by writing down the dual program and verifying Slater's condition.

\begin{lemma}[{\cite[see, e.g.,][Appendix B]{singh2014entropy}}]
\label{lem:KL-dual}
Consider a homogeneous distribution $\mu: \binom{[n]}{k} \to \R_{\geq 0}$ and let  $g_{\mu}(z_1, \dots, z_n)$ be its multivariate generating polynomial. Then, for any $q \in \R^{n}_{\geq 0}$ with $\sum_{i=q}^{n}q_i = 1$, we have
\[ \inf\set*{\DKL{\nu \river \mu} \given \nu D_{k\to 1}=q } = -\log\parens*{\inf_{z_1,\dots, z_n>0} \frac{g_\mu(z_1,\dots,z_n)}{z_1^{kq_1}\cdots z_n^{kq_n}} }. \]
\end{lemma}
\subsection{Correlation, influence, and Dobrushin matrices}
The following matrices are intimately related to the Hessian (at the all-ones vector) of the log-generating polynomial of a distribution on subsets of $[n]$.
\begin{definition}[Signed pairwise influence/correlation matrix] \label{def:corr}
	Let $\mu$ be a probability distribution over $2^{[n]} $ with generating polynomial $g_{\mu} (z_1, \cdots, z_n) = \sum _{S \in  2^{[n]}}\mu(S) z^S.$
	
	Let the \textit{signed pairwise influence matrix} $\inflMat_{\mu} \in \R^{n\times n}$  be defined by
\[
\inflMat_{\mu} (i,j)  = \begin{cases}  0 &\text{ if } j=i \\ \P{j\given i} - \P{j \given \bar{i}} &\text{ else} \end{cases} \]
where $\P{j\given i} = \P_{T \sim \mu}{ j \in T \given i \in T}, \P{j} = \P_{T \sim \mu}{j \in T} $ and $ \P{j \given \bar{i}} = \P_{T \sim \mu}{ j \in T \given i \not\in T}.$

Let the \textit{correlation matrix} $\corMat_{\mu} \in \R^{n\times n}$  be defined by
\[\corMat_{\mu} (i,j) = \begin{cases} 1 -\P{i} &\text{ if } j=i \\ \P{j \given i} - \P{j} &\text{ else}\end{cases} \]
\end{definition}

In our proof of the (non-uniform) fractional log-concavity of the rank-1 Ising model, we will also make use of the well-known Dobrushin matrix. 
\begin{definition}[Dobrushin matrix]
\label{def:dobrushin}
Let $\Omega$ be a finite set and let $\mu$ denote a probability measure on $\Omega^{[n]}$. The Dobrushin matrix $R \in \R^{n\times n}$ is defined by $R_{i,i} = 0$ and for $i\neq j$, 
\[
R_{ij} = 
\max\set*{d_\TV\parens*{\mu(X_j = \cdot \mid X_{-\{i,j\}} = \sigma, X_{i} = \omega_1), \mu(X_j = \cdot \mid X_{-\{i,j\}} = \sigma, X_{i} = \omega_2)} \given \begin{aligned}&\omega_1, \omega_2 \in \Omega,\\ &\sigma \in \Omega^{[n]\setminus \{i,j\}}\end{aligned}}.
\]
\end{definition}

\subsection{Needle decomposition of Ising measures} \label{sec: reduce to rank 1}
The following result, due to Eldan, Koehler, and Zeitouni \cite{eldan2020spectral}, shows that it suffices to prove part (a) of \cref{thm:main} only for rank-1 interaction matrices $J$.  
\begin{theorem}[Needle decomposition of Ising measures \cite{eldan2020spectral}]
\label{thm:needle-decomposition}
Suppose the measure $\mu$ on $\{\pm 1\}^n$ is given by
\[ \mu(x) = \frac{1}{Z} \exp\parens*{\frac{1}{2} \dotprod{x, J x} + \dotprod{h, x}}, \]
i.e.~$\mu$ is an Ising model, and suppose (without loss of generality) that $0 \preceq J \preceq \norm{J}_\OP I$. Suppose also that $f : \{\pm 1\}^n \to \mathbb{R}$ is arbitrary. There exists a mixture decomposition
\[ \mu(x) = \int \mu_{u,v}(x) d\nu(u,v) \]
where $\nu$ is a probability measure on $\mathbb{R}^{2n}$ such that:
\begin{enumerate}
    \item $\nu$-almost surely, $\mu_{u,v}$ is a probabililty measure of the form
    \[ \mu_{u,v}(x) = \frac{1}{Z_{u,v}} \exp\parens*{\frac{1}{2} \dotprod{u, x}^2 + \dotprod{v, x}}, \]
    i.e., a rank one Ising model (``needle''). Furthermore $\norm{u}_2 \le \norm{J}_\OP$.
    \item For all $x \sim y$, we have the following inequality of conductances
    \[ \E*_{\nu}{\frac{\mu_{u,v}(x)\mu_{u,v}(y)}{\mu_{u,v}(x) + \mu_{u,v}(y)}} \le \frac{\mu(x)\mu(y)}{\mu(x) + \mu(y)}. \]
    \item $\E_{\mu_{u,v}}{f(X)} = \E_{\mu}{f(X)}$ $\nu$-almost surely.
\end{enumerate}
\end{theorem}

\begin{corollary}
\label{cor:needle}
Let $n\geq 1$ and suppose that the Glauber dynamics of all rank one Ising models $\mu_{u,v}$ on $\{\pm 1\}^{n}$ with $\norm{u}_2 \le \norm{J}_\OP$ have modified log-Sobolev constant at least $C$. Then, the Glauber dynamics of the Ising model $\mu$ with interaction matrix $J$ and arbitrary external field $h \in \R^{n}$ also has modified log-Sobolev constant at least $C$. 
\end{corollary}

\begin{proof}
As is well-known (cf.~Lemma 13.11 of \cite{levin2017markov}), the Dirichlet form of the Glauber dynamics on $\mu$ has a convenient expression in terms of the edge conductances
\[ \mathcal{E}_{\mu}(f,g) = \sum_{x \sim y} \frac{\mu(x)\mu(y)}{\mu(x) + \mu(y)} (f(x) - f(y))(g(x) - g(y)),\]
where the sum ranges over all adjacent vertices of the hypercube $\{\pm 1\}^n$.
Given this and \cref{thm:needle-decomposition}, we see that if the class of needle measures with $\norm{u}_2 \le \norm{J}_\OP$ has modified log-Sobolev constant at least $C$, then for any $f: \set{\pm 1}^{n} \to \R_{> 0}$ with $\E_{\mu}{f}=1$,
\begin{align*} 
\Ent_{\mu}[f] = \E_{\mu}{f \log f}
&= \E_{\nu}{\E_{\mu_{u,v}}{f \log f}} \\
&= \E_{\nu}{\Ent_{\mu_{u,v}}[f]} \\
&\le \frac{1}{2C}\E_{\nu}{\mathcal{E}_{\mu_{u,v}}(\log f, f)} \\
&= \frac{1}{2C} \sum_{x \sim y} \E*_{\nu}{\frac{\mu_{u,v}(x)\mu_{u,v}(y)}{\mu_{u,v}(x) + \mu_{u,v}(y)}} (\log f(x) - \log f(y))(f(x) - f(y)) \\
&\le \frac{1}{2C} \sum_{x \sim y} \frac{\mu(x)\mu(y)}{\mu(x) + \mu(y)} (\log f(x) - \log f(y))(f(x) - f(y)) \\
&= \frac{1}{2C} \mathcal{E}_{\mu}(\log f, f)
\end{align*}
where in the second line, we used $\E_{\mu_{u,v}}{f} = \E_{\mu}{f} = 1$,  
in the third line we used the assumed MLSI for needle measures, and in the fifth line we used the comparison inequality for conductances.
\end{proof}
	\section{Fractional log-concavity and entropic independence}\label{sec:entropic-independence}

We now prove our main result, \cref{thm:entropic-independence}. Given this result, \cref{thm:alpha frac LC MLSI} follows from a version of the local-to-global argument, as discussed earlier; the proof is given in the Appendix.
	\begin{proof}[Proof of \cref{thm:entropic-independence}]
	Let $\alpha \in (0,1]$ and let
	\[\dagger: \forall (z_1,\dots,z_n) \in \R^{n}_{\geq 0}, \text{ }g_{\mu}(z_1,\dots, z_n) \leq \left(\sum_{i=1}^{n}p_i z_i^{1/\alpha}\right)^{\alpha k} \]
	be the condition appearing in the statement of the theorem. 
	
	We first show that $\alpha$-FLC implies ($\dagger$). Let $\mu$ be an arbitrary $\alpha$-FLC distribution on $\binom{[n]}{k}$ and let $p:=\mu D_{k\to 1} \in \R^n_{\ge 0}$. Note that
		\[ \partial_i g_\mu (1,\dots,1)=\frac{\sum_{S\ni i}\mu(S)}{\sum_{S}\mu(S)}=kp_i. \]
		Since $g_{\mu}(z_1^\alpha,\dots,z_n^\alpha)$ is $\alpha k$-homogeneous and log-concave as a function of $z_1, \dots, z_n$ over $\R_{\geq 0}^n$, it immediately follows that  $f(z_1,\dots,z_n):=g_\mu(z_1^\alpha,\dots,z_n^\alpha)^{1/\alpha k}$ is $1$-homogeneous and quasi-concave and hence, concave (\cref{lem:lc-equiv}). In particular,
		\[ \forall z_1,\dots,z_n>0: f(z_1,\dots,z_n)\leq f(1,\dots,1)+\sum_{i=1}^{n}\partial_i f(1,\dots,1)(z_i-1). \]
		Since $1$-homogeneity of $f$ implies that $\sum_{i=1}^{n}\partial_i f(1,\dots,1)=f(1,\dots,1)$, we see that
		\[\forall z_1,\dots,z_n>0: f(z_1,\dots,z_n)\leq \sum_{i=1}^{n} \partial_i f(1,\dots,1)z_i. \]
		Moreover, since
		\[ \partial_i f(1,\dots,1)= \parens*{\alpha\cdot \partial_i g_\mu(1,\dots,1)}\cdot\parens*{\frac{1}{\alpha k}\cdot g_\mu(1,\dots,1)^{1/\alpha k-1}}=p_i, \]
		we get that
		\[\forall z_1,\dots,z_n>0: f(z_1,\dots,z_n)\leq \sum_{i} p_iz_i. \]
		Rewriting this in terms of $g_\mu$ yields
		\[\forall z_1,\dots, z_n > 0: g_\mu(z_1,\dots, z_n)\leq \parens*{\sum_{i}p_i z_i^{1/\alpha}}^{\alpha k}, \]
	and now, ($\dagger$) follows by continuity. 
	
	Next, we show that ($\dagger$) implies $(1/\alpha)$-entropic independence. Let $\nu$ be an arbitrary distribution on $\binom{[n]}{k}$ and let $q:= \nu D_{k \to 1}$, so that $q \in \R^{n}_{\geq 0}$ with $\sum_{i=1}^{n}q_i = 1$. We have from \cref{lem:KL-dual} that
		\[ \DKL{\nu \river \mu} \geq \inf\set*{\DKL{\nu \river \mu} \given \nu D_{k\to 1}=q } = -\log\parens*{\inf_{z_1,\dots, z_n>0} \frac{g_\mu(z_1,\dots,z_n)}{z_1^{kq_1}\cdots z_n^{kq_n}} }. \]
        By ($\dagger$), 	
		\[ \inf_{z_1,\dots, z_n>0} \frac{g_\mu(z_1,\dots,z_n)}{z_1^{kq_1}\cdots z_n^{kq_n} }\leq \inf_{z_1,\dots,z_n>0}\frac{\parens*{\sum_i p_i z_i^{1/\alpha}}^{\alpha k} }{z_1^{kq_1}\cdots z_n^{kq_n}}. \]
		Plugging in $z_i=(q_i/p_i)^\alpha$, we obtain
		\[ \inf_{z_1,\dots, z_n>0} \frac{g_\mu(z_1,\dots,z_n)}{z_1^{kq_1}\cdots z_n^{kq_n} }\leq \prod_{i=1}^n (p_i/q_i)^{\alpha k q_i}. \]
		Taking $\log$ and negating gives 
		\[ \DKL{\nu \river \mu}\geq -\log \prod_{i=1}^n (p_i/q_i)^{\alpha k q_i}=\alpha k \sum_i q_i \log(q_i/p_i)=\alpha k\cdot  \DKL{\nu D_{k\to 1}\river \mu D_{k\to 1}}. \]
		Since $\nu$ is arbitrary, we obtain $(1/\alpha)$-entropic independence. 
		
		Now, we show that $(1/\alpha)$-entropic independence implies ($\dagger$). By induction on the lower-dimensional faces of the positive orthant and homogeneity, it suffices to show that
		\[\l(z_1,\dots, z_n) := \left(\sum_{i}p_i z_i^{1/\alpha}\right)^{\alpha k}\geq 1 \quad \forall z = (z_1,\dots,z_n) \in C,\]
		where $C = \{(z_1,\dots,z_n): g_{\mu}(z_1,\dots,z_n) = 1 \wedge  \nabla (\l/g_{\mu})(z_1,\dots,z_n) = 0\}$. Let $z^* = (z_1^*,\dots, z_n^*) \in C$. Then, $\nu = z^* \ast \mu$ is a distribution on $\binom{[n]}{k}$ with $q = (q_1,\dots, q_n) := \nu D_{k \to 1} \propto (p_1 (z_1^*)^{1/\alpha},\dots, p_n (z_n^*)^{1/\alpha})$. Now, examining the first order condition shows that
		\[\DKL{\nu \river \mu} = -\log\left( \inf_{z_1,\dots, z_n>0} \frac{g_\mu(z_1,\dots,z_n)}{z_1^{kq_1}\cdots z_n^{kq_n} }\right) = -\log\left(\frac{1}{(z_1^*)^{kq_1}\dots (z_n^*)^{kq_n}}\right),\]
	    so that by $(1/\alpha)$-entropic independence, 
	    \[\frac{1}{(z_1^*)^{kq_1}\dots (z_n^*)^{kq_n}} \leq \prod_{i=1}^n (p_i/q_i)^{\alpha k q_i} = \frac{\l(z_1^*,\dots, z_n^*)}{(z_1^*)^{kq_1}\dots (z_n^*)^{kq_n}},\]
	    from which we get that $\l(z_1^*,\dots, z_n^*) \geq 1$, as desired. 
	    
	    Finally, we establish the equivalence between entropic independence under arbitrary external fields and fractional log-concavity. In one direction, we note that $\alpha$-fractional log-concavity of $\mu$ immediately implies  $\alpha$-fractional log-concavity of $\lambda \ast \mu$ for any $\lambda = (\lambda_1,\dots, \lambda_n) \in \R^{n}_{>0}$, which as we have just seen, implies $(1/\alpha)$-entropic independence of $\lambda \ast \mu$. In the other direction, suppose that $\lambda \ast \mu$ is $(1/\alpha)$-entropic independent for all $\lambda = (\lambda_1,\dots, \lambda_n) \in \R^{n}_{>0}$. Then, using ($\dagger$) for all $\lambda \ast \mu$, we see that
	    \[\forall z_1,\dots, z_n > 0: g_{\mu}(z_1^\alpha, \dots, z_n^\alpha)^{1/\alpha k} = \inf_{\lambda \in \R^{n}_{>0}}\sum_{i}p(\lambda)_i z_i.\]
	    Since a pointwise infimum of linear functions is concave, it follows that $g_{\mu}(z_1^{\alpha},\dots, z_n^{\alpha})$ is $\alpha k$-root-concave, and hence, log-concave (\cref{lem:lc-equiv}). This completes the proof. 
	\end{proof}
	

	In the course of our proof of \cref{thm:main}, we will establish that rank-1 Ising measures $\mu_{u,h}$ with $\|u\|_{2} < 1$ are non-uniformly fractionally log-concave. The following proposition shows that non-uniform fractional log-concavity leads to a corresponding non-uniform version of entropic independence.
	\begin{proposition} \label{thm:nonuniform FLC to entropy contraction}
	    Let $\mu$ be a distribution on $\set{\pm 1}^{[n]}$ and suppose that there exist $\alpha_1,\dots, \alpha_n \in (0,1]$ for which the distribution $\mu^{\hom}$ on $\binom{[n] \cup [\bar{n}]}{n}$ is $(\alpha_1, \dots, \alpha_n,\alpha_1, \dots, \alpha_n)$-FLC. 
	    
	    Let $\nu$ be a distribution on $\{\pm 1\}^{n}$ and for $i \in [n]$, let $\mu_i$ (respectively $\nu_i$) denote the marginal distribution of the $i^{th}$ coordinate under $\mu$ (respectively $\nu$). Then,
	    \[\sum_{i=1}^{n}\alpha_{i}\DKL{\nu_i \river \mu_i} \leq \DKL{\nu \river \mu}.\]
	\end{proposition}
		The proof of this proposition is similar to \cref{thm:entropic-independence} and is deferred to \cref{sec:appendix}.
	
	\section{MLSI for rank-1 Ising measures}
\label{sec:MLSI-Ising}
For $u,h \in \mathbb{R}^{n}$, let $\mu = \mu_{u,h}$ denote the probability distribution on the discrete hypercube $\{-1,1\}^{n}$ given by
\[\mu(x) \propto \exp\left(\frac{1}{2}\langle u, x\rangle ^{2} + \langle h, x \rangle\right) \quad \forall x \in \{-1,1\}^{n}\]
and let $\nu$ denote an arbitrary probability distribution on $\{-1,1\}^{n}$. Viewing $\mu$ and $\nu$ as homogeneous distributions on $\binom{[2n]}{n}$, we may consider the homogeneous distributions $\mu D$ and $\nu D$ on $\binom{[2n]}{n-1}$, where $D := D_{n \to (n-1)}$ is the down operator.

The following key result shows that the modified log-Sobolev constant of the Glauber dynamics of $\mu_{u,h}$ is at least $2(1-\|u\|_{2}^{2})/n$. Together with \cref{cor:needle}, it immediately implies \cref{thm:main}\ref{part:mlsi}.

\begin{proposition} \label{prop:rank 1 MLSI}
For any $n\geq 2$, any $u,h \in \mathbb{R}^{n}$, and for any probability measure $\nu$ on $\{-1,1\}^{n}$,
\[\DKL{\nu P_{\mu_{u,h}} \river \mu_{u,h} P_{\mu_{u,h}}} \leq \DKL{\nu D \river \mu_{u,h} D} \leq \left(1 - \frac{1-\|u\|_2^2}{n}\right)\DKL{\nu \river \mu_{u,h}}.\]

In particular, by \cref{lem:entropy-contraction-implies-mlsi}, 
\[\rho_0(P_{\mu_{u,h}}) \geq \frac{2(1-\|u\|_{2}^{2})}{n}.\]
\end{proposition}

We will prove this result assuming the following (non-uniform) fractional log-concavity of rank-1 Ising models, which will be proved in the next subsection. 

\begin{proposition} \label{prop:frac LC}
For all $n\geq 2$ and $u,h \in \mathbb{R}^{n}$ with $\|u\|_{2}\leq 1$, $g_{\mu_{u,h}}$ is $(\alpha_1,\dots, \alpha_n, \alpha_1,\dots, \alpha_n)$-log-concave on $\mathbb{R}^{2n}_{\geq 0}$, where $\alpha_i = (1-\|u_{-i}\|_2^2)$.
\end{proposition}

\begin{proof}[Proof of \cref{prop:rank 1 MLSI}]
The first inequality is simply the data-processing inequality. We proceed to prove the second inequality. 
For convenience, we will denote $\mu_{u, h}$ simply by $\mu$. For $i \in [n]$ and $\epsilon_i \in \{-1,1\}$, we let $p_{i,\epsilon_i} = \mu(X_i = \epsilon_i)$ and $q_{i,\epsilon_i} = \nu(X_i = \epsilon_i)$. Also, let $f(x) = x\log{x}$. 

We will prove the assertion by induction on $n \geq 2$. We begin by establishing the base case $n=2$. By the non-negativity of the KL divergence, we have for all $i \in [2]$ and $\epsilon_{i} \in \{-1,1\}$ that 
\begin{align*}
    &\frac{1}{q_{i,\epsilon_i}}\cdot u_{-i}^{2}\cdot \sum_{x \in \{-1,1\}}\mu(X_j = x, X_i = \epsilon_i) \left(f\left(\frac{\nu(X_j = x, X_i = \epsilon_i)}{\mu(X_j = x, X_i = \epsilon_i)}\right) - f\left(\frac{q_{i,\epsilon_i}}{p_{i,\epsilon_i}}\right)\right)\\
    &= u_{-i}^{2}\DKL{\nu (\cdot \mid X_i = \epsilon_i) \river \mu (\cdot \mid X_i = \epsilon_i)}\\
    &\geq 0\\
    &= \frac{1}{q_{i,\epsilon_i}}p_{i,\epsilon_i}\left(f\left(\frac{q_{i,\epsilon_i}}{p_{i,\epsilon_i}}\right) - f\left(\frac{q_{i,\epsilon_i}}{p_{i,\epsilon_i}}\right)\right),
\end{align*}
which may be rewritten as
\begin{align*}
 u_{-i}^{2}\cdot \sum_{x \in \{-1,1\}}\mu(X_j = x, X_i = \epsilon_i)f\left(\frac{\nu(X_j = x, X_i = \epsilon_i)}{\mu(X_j = x, X_i = \epsilon_i)}\right) + (1-u_{-i}^2) p_{i,\epsilon_i} f\left(\frac{q_{i,\epsilon_i}}{p_{i,\epsilon_i}}\right)
 \geq p_{i,\epsilon_i}f\left(\frac{q_{i,\epsilon_i}}{p_{i,\epsilon_i}}\right).
\end{align*}
Summing this over $\epsilon_i \in \{\pm 1\}$ and averaging over $i \in [2]$, we get that
\begin{align*}
    \frac{\|u\|_{2}^{2}}{2} \DKL{\nu \river \mu } + \frac{1}{2}\sum_{i=1}^{2}(1-u_{-i}^2) \sum_{\epsilon_{i} \in \{-1.1\}}p_{i,\epsilon_i}\cdot f\left(\frac{q_{i,\epsilon_i}}{p_{i,\epsilon_i}}\right) &\geq \frac{1}{2}\sum_{i=1}^{2} \sum_{\epsilon_{i} \in \{-1.1\}}p_{i,\epsilon_i}\cdot f\left(\frac{q_{i,\epsilon_i}}{p_{i,\epsilon_i}}\right)\\
    &= \DKL{\nu D \river \mu D}.
\end{align*}
Using \cref{prop:frac LC,thm:nonuniform FLC to entropy contraction}, this shows that
\begin{align*}
    \left(\frac{1 + \|u\|_{2}^{2}}{2}\right)\DKL{\nu \river \mu} \geq \DKL{\nu D \river \mu D},
\end{align*}
which is equivalent to the claimed statement for $n=2$. 

For the inductive step, suppose that the assertion has already been established for some $n-1\geq 2$ and for all $u',h' \in \mathbb{R}^{n-1}$. By a similar argument to the one above, we show how to establish the assertion for $n$ and for any given $u,h \in \mathbb{R}^{n}$. 

We begin by noting that for any $i \in [n]$ and for any $\epsilon_i \in \{\pm 1\}$, the probability measure $\mu_{u,h}(\cdot \mid  X_i = \epsilon_i)$ is of the form $\mu_{u_{-i},h'}$ for some $h'\in \mathbb{R}^{n-1}$. Therefore, by the inductive hypothesis, for all $i \in [n]$ and for all $\epsilon_i \in \{\pm 1\}$, we have that
\begin{equation}
\label{eqn:ih}
\DKL{\nu(\cdot \mid X_i = \epsilon_i) D \river \mu_{u,h} (\cdot \mid X_i = \epsilon_i)D} \leq \left(1 - \frac{1-\|u_{-i}\|_2^2}{n-1}\right)\DKL{\nu(\cdot \mid X_i = \epsilon_i) \river \mu_{u,h}(\cdot \mid X_i = \epsilon_i)}.
\end{equation}

We rewrite this more explicitly. We denote $\mu_{u,h}$ simply by $\mu$. Given $i\in [n]$, $\epsilon_i \in \{\pm 1\}$ and $x \in \{-1,1\}^{n-1}$, let $x \oplus \epsilon_i \in \{-1,1\}^{n}$ denote the element with $\epsilon_i$ occupying the $i^{th}$ coordinate and the elements of $x$ occupying the remaining coordinates (in the natural order). For $i\neq k$ and $x \in \{-1,1\}^{n-2}$, we define $x \oplus \epsilon_i \oplus \epsilon_k$ similarly. 
Then, the right-hand side of \cref{eqn:ih} is equal to
\begin{align*}
    \frac{1}{q_{i,\epsilon_i}}\cdot\left(1 - \frac{1-\|u_{-i}\|_2^2}{n-1}\right)\cdot \sum_{x \in \{-1,1\}^{n-1}}\mu(x \oplus \epsilon_i)\left(f\left(\frac{\nu(x\oplus \epsilon_i)}{\mu(x\oplus \epsilon_i)}\right) - f\left(\frac{q_{i,\epsilon_i}}{p_{i,\epsilon_i}}\right)\right).
\end{align*}
For the left-hand side, first note that for any $k \neq i$ and for any $x \in \{-1,1\}^{n}$,
\begin{align*}
    (\mu(\cdot \mid X_i = \epsilon_i)D)(X_j = x_j \,\, \forall j\notin \{i,k\}) 
    &= \frac{1}{n-1}\sum_{\epsilon_k \in \{\pm 1\}}\mu(x_{-\{i,k\}}\oplus \epsilon_k \mid X_i = \epsilon_i) \\
    &= \frac{1}{n-1}\sum_{\epsilon_k \in \{\pm 1\}}\frac{\mu(x_{-\{i,k\}}\oplus \epsilon_k \oplus \epsilon_i)}{p_{i,\epsilon_i}}
\end{align*}
and similarly for $\nu(\cdot \mid X_i = \epsilon_i)D$. For $x_{-\{i,k\}} \in \{-1,1\}^{n-2}$ and $x_{-k} \in \{-1,1\}^{n-1}$, denote 
\[\mu(x_{-\{i,k\}}\oplus\epsilon_i) = \sum_{\epsilon_k \in \{\pm 1\}} \mu(x_{-\{i,k\}}\oplus \epsilon_k \oplus \epsilon_i), \quad \quad \mu(x_{-k}) = \sum_{\epsilon_k \in \{-1,1\}}\mu(x_{-k} \oplus \epsilon_k)\]
and similarly for $\nu$. Then, the left-hand side of \cref{eqn:ih} is equal to
\begin{align*}
    \frac{1}{q_{i,\epsilon_i}\cdot (n-1)}\sum_{k\neq i}\,\,\,\sum_{x_{-\{i,k\}} \in \{-1,1\}^{n-2}}\mu(x_{-\{i,k\}}\oplus \epsilon_i)\left(f\left(\frac{\nu(x_{-\{i,k\}}\oplus \epsilon_i)}{\mu(x_{-\{i,k\}}\oplus \epsilon_i)}\right) - f\left(\frac{q_{i,\epsilon_i}}{p_{i,\epsilon_i}}\right)\right).
\end{align*}
Therefore, we can rewrite \cref{eqn:ih} as 
\begin{align}
\label{eqn:ih-explicit}
 &\left(1 - \frac{1-\|u_{-i}\|_2^2}{n-1}\right)\cdot \sum_{x \in \{-1,1\}^{n-1}}\mu(x \oplus \epsilon_i)f\left(\frac{\nu(x\oplus \epsilon_i)}{\mu(x\oplus \epsilon_i)}\right) + \frac{1-\|u_{-i}\|_2^2}{n-1}\cdot p_{i,\epsilon_i}\cdot f\left(\frac{q_{i,\epsilon_i}}{p_{i,\epsilon_i}}\right) \nonumber \\
 &\geq  \frac{1}{n-1}\sum_{k\neq i}\,\,\,\sum_{x_{-\{i,k\}} \in \{-1,1\}^{n-2}}\mu(x_{-\{i,k\}}\oplus \epsilon_i)\cdot f\left(\frac{\nu(x_{-\{i,k\}}\oplus \epsilon_i)}{\mu(x_{-\{i,k\}}\oplus \epsilon_i)}\right)
\end{align}
Summing \cref{eqn:ih-explicit} over $\epsilon_i \in \{\pm 1\}$ and then averaging over $i \in [n]$, we get that
\begin{align*}
    &\left(1 - \frac{1}{n-1} + \frac{\|u\|_{2}^{2}}{n}\right)\sum_{x \in \{-1,1\}^{n}}\mu(x)f\left(\frac{\nu(x)}{\mu(x)}\right) + \frac{1}{n}\sum_{i=1}^{n} \frac{1-\|u_{-i}\|_2^2}{n-1} \sum_{\epsilon_{i} \in \{\pm 1\}}  p_{i,\epsilon_i}\cdot f\left(\frac{q_{i,\epsilon_i}}{p_{i,\epsilon_i}}\right) \\
    &\geq \frac{1}{n}\sum_{i=1}^{n}\frac{1}{n-1}\sum_{k\neq i}\sum_{x_{-k} \in \{-1,1\}^{n-1}}\mu(x_{-k})\cdot f\left(\frac{\nu(x_{-k})}{\mu(x_{-k})}\right)\\
    &= \frac{1}{n}\sum_{k=1}^{n} \sum_{x_{-k} \in \{-1,1\}^{n-1}}\mu(x_{-k})\cdot f\left(\frac{\nu(x_{-k})}{\mu(x_{-k})}\right) \\
    &= \DKL{\nu D \river \mu D}.
\end{align*}

Finally, by \cref{prop:frac LC,thm:nonuniform FLC to entropy contraction}, it follows that
\begin{align*}
    &\left(1 - \frac{1-\|u\|_2^2}{n}\right)\DKL{\nu \river \mu}
    = \left(1 - \frac{1}{n-1} + \frac{\|u\|_{2}^{2}}{n}\right)\DKL{\nu \river \mu} + \frac{1}{n(n-1)}\DKL{\nu \river \mu}\\
    &\geq \left(1 - \frac{1}{n-1} + \frac{\|u\|_{2}^{2}}{n}\right)\sum_{x \in \{-1,1\}^{n}}\mu(x)f\left(\frac{\nu(x)}{\mu(x)}\right) + \frac{1}{n}\sum_{i=1}^{n} \frac{1-\|u_{-i}\|_2^2}{n-1} \sum_{\epsilon_{i} \in \{\pm 1\}}  p_{i,\epsilon_i}\cdot f\left(\frac{q_{i,\epsilon_i}}{p_{i,\epsilon_i}}\right)\\
    &\geq \DKL{\nu D \river \mu D},
\end{align*}
which completes the inductive step. 
\end{proof}

	\subsection{Non-uniform fractional log-concavity of rank-1 Ising measures}


We now turn to the proof of \cref{prop:frac LC}. 

\begin{proof}[Proof of \cref{prop:frac LC}]
For notational convenience, we will denote $g_{\mu}$ simply by $g$. Without loss of generality, we may assume that $u \in \R^{n}_{> 0}$.  
We wish to show that for any $u \in \R^{n}_{> 0}, h \in \mathbb{R}^{n}$,
\[\log g(z_1^{\alpha_1}, \dots, z_n^{\alpha_n}, z_{\bar{1}}^{\alpha_1}, \dots, z_{\bar{n}}^{\alpha_n})\] is concave for all $\vec{x} \in \R_{\geq 0}^{2n}$. It suffices to show this for the point $\vec{1} = (1,1,\dots, 1) \in \mathbb{R}^{2n}$. Indeed, for an arbitrary point $\vec{x} = (e^{y_1},\dots, e^{y_n},e^{y_1'},\dots, e^{y_n'}) \in \mathbb{R}^{2n}_{>0}$, where $y_1,\dots, y_n, y_1',\dots, y_n' \in \R$, note that
\begin{align*}
g(z_1e^{\alpha _1 y_1},\dots, z_n e^{\alpha _n y_n}, z_{\bar{1}}e^{\alpha_1 y_1'},\dots, z_{\bar{n}}e^{\alpha_n y_n'})
&= C_{\alpha_1,\dots, \alpha_n}(y_1,\dots, y_n, y_1',\dots, y_n')g_{\mu_{u,h'}}(z_1,\dots, z_n, z_{\bar{1}},\dots, z_{\bar{n}})
\end{align*}
for some $h' \in \R^{n}$. Therefore, 
\[\nabla^2\log g(z_1^{\alpha_1},\dots, z_{\bar{n}}^{\alpha_n}) \mid_{\vec{x}} = \diag(v_1^{-1},\dots, v_{2n}^{-1})\cdot \nabla^{2}\log g_{\mu_{u, h'}}(z_1^{\alpha_1},\dots, z_{\bar{n}}^{\alpha_n}) \mid_{\vec{1}} \cdot \diag(v_1^{-1},\dots, v_{2n}^{-1}) \preceq 0.\]

Now, let $H = \nabla^2 \log g(z_1^{\alpha_1},\dots, z_{\bar{n}}^{\alpha_n})|_{\vec{1}}$. By direct computation, we see that
\[
H_{ij} = 
\begin{cases}
\alpha_i(\alpha_i-1)\mathbb{P}_{\mu}[i] - \alpha_i^2\mathbb{P}_{\mu}[i]^2 \quad \quad &\text{if } j=i\\
\alpha_i \alpha_j \left(\mathbb{P}_{\mu}[i \wedge j] - \mathbb{P}_{\mu}[i]\mathbb{P}_{\mu}[j]\right) \quad \quad &\text{otherwise}.
\end{cases}
\]
Our goal is to show that $H \preceq 0$. Let $D_{\vec{\alpha}} = \diag(\alpha_1,\dots, \alpha_{n}, \alpha_1,\dots, \alpha_{n}) \in \R^{2n\times 2n}$ and let $D_{\mu} \in \R^{2n \times 2n}$ be the diagonal matrix with $(D_{\mu})_{i,i} = \mathbb{P}_{\mu}[i]^{-1}$. Note that
\[\corMat_{\mu} = D_{\vec{\alpha}}^{-1} D_{\mu} H D_{\vec{\alpha}}^{-1} + D_{\vec{\alpha}}^{-1}.\]
Moreover, from the proof of \cite[Lemma~69]{alimohammadi2021fractionally}, it follows that the non-zero eigenvalues of $\corMat_{\mu}D_{\vec{\alpha}}$ coincide with the non-zero eigenvalues of 
 \[\diag(\alpha_1,\dots, \alpha_n)(\inflMat_\mu + I_{n\times n}) \in \R^{n\times n} =: A\]
Therefore,
\[\lambda_{\max}(A) \leq 1 \implies \lambda_{\max}(D_{\vec{\alpha}}^{-1}D_{\mu}H) \leq 0 \implies H \preceq 0.\]

Let $\|\cdot \|_{u}$ be the norm on $\R^{n}$ defined by $\|\vec{x}\|_{u} = \max_{i \in [n]}|u_i^{-1} x_i|$. In order to show that $\lambda_{\max}(A) \leq 1$, it suffices to show that $\sup\{\|A\vec{x}\|_{u} : \|\vec{x}\|_{u} \leq 1\}\leq 1$. Explicitly, this amounts to showing that 
\begin{align*}
    \sum_{j=1}^{n}|A|_{ij}u_{j} \leq u_{i} \quad \forall i \in [n],
\end{align*}
which is implied by (recall that $u \in \R^{n}_{> 0}$)
\begin{align*}
\sum_{j=1}^{n}\alpha_i|\inflMat_{\mu}(i,j)|u_j + \alpha_i u_i \leq u_i \quad \forall i \in [n].
\end{align*}

We will show that for all $i \in [n]$,
\begin{equation}
\label{ineq:main}
    \sum_{j\neq i}|\inflMat_{\mu}(i,j)|u_j \leq \alpha_i^{-1}u_i \|u_{-i}\|_{2}^{2}.
\end{equation}

Given this, we see that for all $i \in [n]$,
\begin{align*}
   \sum_{j=1}^{n}\alpha_i|\inflMat_{\mu}(i,j)|u_j + \alpha_i u_i \leq u_i \leq  u_i (\|u_{-i}\|_{2}^{2}+\alpha_i) = u_i, 
\end{align*}
as desired. 

The remainder of the proof is devoted to proving \cref{ineq:main}. Without loss of generality, it suffices to consider the case $i=1$. Consider the probability measures
\[\mu_+ = \mu(\cdot \mid X_1 = 1), \quad \mu_{-} = \mu(\cdot \mid X_1 = -1)\]
on $\{-1,1\}^{[n]\setminus \{1\}}$. Observe that for any $i \neq j \in [n]\setminus \{1\}$, the entry $R_{ij}$ for the Dobrushin matrix of $\mu_+$ (\cref{def:dobrushin}) satisfies
\begin{align*}
    |R_{ij}| 
    &\leq \left|\frac{\exp(u_i u_j)}{\exp(u_i u_j) + \exp(-u_i u_j)} - \frac{\exp(-u_i u_j)}{\exp(u_i u_j) + \exp(-u_i u_j)}\right| = |\tanh(u_i u_j)| \leq |u_i u_j| = u_i u_j,
\end{align*}
where we have used the 1-Lipschitzness of $\tanh$ and the assumption that $u \in \R^{n}_{>0}$. In particular, we have for all $i \in [n]\setminus \{1\}$ that
\begin{align*}
    \sum_{j=2}^{n}R_{ij}u_j \leq \sum_{j=2}^{n}|R_{ij}|u_j \leq u_i\sum_{j=2}^{n}u_j^{2} = u_i\|u_{-1}\|_{2}^{2} = u_i(1-\alpha_1). 
\end{align*}
By \cite[Lemma~20]{dyer2009matrix}, this shows that $\mu_{1}$ is $(1-\alpha_1/(n-1))$-contractive with respect to the Glauber dynamics and the weighted Hamming metric $d_{u}$ on $\{-1,1\}^{[n]\setminus \{1\}}$, defined by $d_{u}(\sigma, \tau) = \sum_{i=2}^{n}u_i \mathds{1}[\sigma_i \neq \tau_i]$.

Given this, we can use the following lemma (which immediately follows by applying \cite[Lemma~4.3]{blanca2021mixing} to the function $f(x) = \sum_{i=1}^{n}w_i \mathds{1}(x_i = 1)$) to obtain an initial upper bound on the left hand side of \cref{ineq:main}. 
\begin{lemma}[\cite{blanca2021mixing}] \label{lem:weight norm}
Let $\mu, \nu$ be two distributions on $\set*{\pm 1}^{n}.$ Let $P_{\mu}$ and $P_{\nu}$ be Markov chains on $\{\pm 1\}^{n}$ with stationary distributions $\mu$ and $\nu$ respectively. Let $w \in \R_{\geq 0}^{n}$ and let $d_{w}$ denote the weighted Hamming distance on $\{-1,1\}^{n}$ corresponding to $w$.
Suppose that $\mu$ is $(1-\epsilon/n)$-contractive with respect to the chain $P_{\mu}$ and the metric $d_{w}$. 
Then,
\begin{equation*}
          L:=\sum_{i=1}^n w_i \abs{\P_{X\sim \mu}{X_i= 1} - \P_{X\sim \nu}{X_i=1}} 
        \leq  \frac{n}{\epsilon}\cdot \mathbb{E}_{\sigma \sim \nu}\left[\sum_{i=1}^n w_i \abs{P_{\mu} (\sigma \to \sigma^{i} ) - P_{\nu} (\sigma \to \sigma^i)  } \right],
\end{equation*}
where $\sigma^i$ is $\sigma $ with the $i$-th coordinate flipped. 
\end{lemma}

Applying this to the measures $\mu_+$ and $\mu_-$ for their respective Glauber dynamics and the weighted Hamming distance $d_u$, we see that
\begin{align*}
    \sum_{j=2}^{n}|\inflMat_{\mu}(1,j)|u_j
    &=  \sum_{j=2}^{n}\left|\mu(X_j = 1 \mid X_1 = 1) - \mu(X_j = 1 \mid X_1 = -1)\right|u_j\\
    &= \sum_{j=2}^{n} u_j \left|\P_{X \sim \mu_+}{X_j = 1} - \P_{X \sim \mu_-}{X_j = 1}\right|\\
    &\leq \frac{n-1}{\alpha_1}\mathbb{E}_{\sigma\sim \mu_-}\left[\sum_{j=2}^n u_j \abs{P_{\mu_+} (\sigma \to \sigma^{j} ) - P_{\mu_-} (\sigma \to \sigma^j)  } \right] \qquad \qquad \qquad \qquad \qquad \textrm{(\cref{lem:weight norm})}\\
    &\leq \frac{n-1}{\alpha_1}\max_{\sigma \in \{-1,1\}^{[n]\setminus \{1\}}}\left[\sum_{j=2}^n u_j \abs{P_{\mu_+} (\sigma \to \sigma^{j} ) - P_{\mu_-} (\sigma \to \sigma^j)  } \right]\\
    &\leq \frac{n-1}{\alpha_1}\cdot \sum_{j=2}^{n}\frac{u_j}{n-1}\cdot \max_{\epsilon \in \{-1,1\}} \left|\left(\frac{1}{2} + \frac{1}{2}\tanh(u_1 u_j\epsilon + h_j\epsilon)\right) - \left(\frac{1}{2} + \frac{1}{2}\tanh(-u_1 u_j\epsilon + h_j\epsilon)\right) \right|\\
    &\leq \frac{n-1}{\alpha_1}\cdot \sum_{j=2}^{n}\frac{u_j}{n-1}\cdot |u_1 u_j|\\
    &= \frac{u_1}{\alpha_1}\sum_{j=2}^{n}u_j^{2} = \alpha_{1}^{-1}u_1\|u_{-1}\|_{2}^{2},
\end{align*}
as desired in \cref{ineq:main}. Here, the penultimate line follows from the 1-Lipschitzness of $\tanh$. \qedhere

\end{proof}

\section{Mixing time for the Ising model using the exchange property}
\label{sec:exchange}
In this section, we prove \cref{thm:main}\ref{part:mixing}. Based on \cref{thm:main}\ref{part:mlsi}, we already know that the Ising model satisfies the modified log-Sobolev inequality, which implies mixing time $O(n \log(n))$ as long as the external field $h$ is not too large, in particular if $h$ is size $poly(n)$. In fact, with a little more care we can eliminate this dependence on $h$ entirely.
To do so, we show that the Ising model $\mu_{J, h}$ satisfies the exchange property in \cite{logConcaveIV} and then use \cite[Lemma~22, 23]{logConcaveIV} to deduce the bound on the mixing time from \cref{thm:main}\ref{part:mlsi}. 
\begin{definition}[Exchange property]
Consider a distribution $\mu$ on $\binom{[n]}{k}$. We say that $\mu$ has the exchange property if for any $S, T  \in \binom{[n]}{k}$ and for any $i\in S \setminus T$, there exists $j\in T \setminus S$ such that
 \[\mu(S) \mu(T) \leq 2^{O(k)} \mu(S-i+j) \mu(T-j+i).\]
\end{definition}
\begin{lemma}[{\cite[Lemma~22, 23]{logConcaveIV}}]
Consider a distribution $\mu$ on $\binom{[n]}{k}$. Suppose that $\mu$ has the exchange property and that the $k \leftrightarrow (k-1)$ down-up walk to sample from $\mu$ has modified log-Sobolev constant at least $\alpha.$ Then, the $k \leftrightarrow (k-1)$ down-up walk starting from any initial state reaches a $2^{O(k^2)}$-warm start after $O(k\log k)$ steps, and hence, mixes in $O\parens*{(k +\alpha^{-1}) \log k}$ steps.  
\end{lemma}
Recall that for an Ising measure $\mu = \mu_{J,h}$ on $\{\pm 1\}^{n}$, the Glauber dynamics for $\mu$ is exactly the $n \leftrightarrow (n-1)$ down-up walk on $\mu^{\hom}: \binom{[n] \cup [\bar{n}]}{n} \to \R_{\geq 0}.$ The following lemma shows that $\mu^{\hom}$ has the exchange property, and hence, by \cref{thm:main}\ref{part:mlsi}, that the  $n \leftrightarrow (n-1)$ down-up walk for $\mu^{\hom}$ i.e. the Glauber dynamics for $\mu$ mixes in $O\parens*{n\log n/(1-\norm{J}_\OP)}$ steps from any starting state. 
\begin{lemma} \label{prop:exchange}
Let $\mu$ be the measure on $\{\pm 1\}^{n}$ defined by 
\[\mu(x) = \frac{1}{Z} \exp\left(\frac{1}{2} \langle x, J x \rangle + \langle h, x \rangle\right)\]
where $\norm{J}_\OP \leq 1.$ Consider $\sigma \neq \tau \in \set*{\pm 1}^n$ and for any $i \in [n]$ for which $\sigma_i \neq \tau_i$, let 
$\sigma', \tau'$ be $\sigma$ and $\tau$ with the $i$-th entry flipped. 
Then,
\[\mu(\sigma) \mu(\tau) \leq 2^{O(\sqrt{n})} \mu(\sigma') \mu(\tau').\]
\end{lemma}
\begin{proof}
Without loss of generality, we may assume that $\sigma_i = 1 = -\tau_i$. It is
easy to see that $\mu(\sigma)/\mu(\sigma') = \exp\left(2\sum_{j\neq i}  J_{ij} \sigma_j + 2 h_i \right) $ and  $\mu(\tau)/\mu(\tau') = \exp\left(-2 \sum_{j\neq i}  J_{ij} \tau_j - 2 h_i \right) .$ Thus
\[ \frac{\mu(\sigma) \mu(\tau)}{\mu(\sigma') \mu(\tau')} \leq \exp\parens*{ 2\sum_{j \neq i} J_{ij} \parens{\sigma_j - \tau_j} } \leq 2^{O(\sqrt{n})},\]
where the last inequality holds since $\abs{\sum_{j\neq i} J_{ij} (\sigma_j - \tau_j)} \leq 2 \norm{J}_{\infty \to \infty} \leq 2\sqrt{n} \| J \|_{2 \to \infty} \leq 2\sqrt{n}\norm{J}_\OP\leq 2\sqrt{n}.$
\end{proof}

	\PrintBibliography
	\appendix
	\appendixpage
	\addappheadtotoc
	\section{Example with negligible low-level entropy contraction}
\label{sec:bad-example}

\begin{example}[An FLC example where $2 \to 1$ contraction fails to hold]\label{example:no-2-to-1}
Let $\mu^{(2)}: \binom{[n]\times [2]}{2k}$ be defined by probability mass function 
\[ \mu^{(2)}(S \times [2])= \frac{1}{\binom{n}{k}} \]
for $S \in \binom{[n]}{k}$, and with $\mu^{(2)}(S') = 0$ for sets $S'$ not of this form. This is the $2$-fold distribution, as in \cref{remark:lowerbound on downup walk}, corresponding to the uniform distribution $\mu$ over $\binom{[n]}{k},$ and so in particular $\mu^{(2)}$ is $1/2$-FLC. Based on our our main result, \cref{thm:entropic-independence}, we know that $\mu^{(2)}$ satisfies entropic independence, i.e. $k \to 1$ entropy contraction, with a constant depending only on $k$. 

However, it does not satisfy nontrivial $2 \to 1$ entropy contraction. Let $a \in [n]$ be arbitrary,  let $a_1,a_2$ denote the two corresponding elements of $[n] \times [2]$, and consider the probability measure $\nu = \delta_{\{a_1,a_2\}}$ on $\binom{[n] \times [2]}{2}$. Then we can compute that $\DKL{\nu \river \mu D_{k \to 2}} = \log(n(2k - 1))$, $\DKL{\nu D_{2 \to 1} \river \mu D_{k \to 1}} = \log(n)$, and so
\[ \DKL{\nu D_{2 \to 1} \river \mu D_{k \to 1}} = \frac{1}{1 + \log(2k - 1)/\log(n)} \DKL{\nu \river \mu D_{k \to 2}}.  \]
We see that as $n \to \infty$ the constant on the right hand side goes to $1$, which corresponds to the trivial upper bound guaranteed by the data processing inequality.
\end{example}

\section{Deferred proofs} 
\label{sec:appendix}
We include the omitted proofs.

To prove \cref{thm:alpha frac LC MLSI}, similar to \cite[Theorem~46]{alimohammadi2021fractionally}, we write $\DKL{\nu \river \mu}$ and $ \DKL{\nu D_{k\to \l} \river \mu D_{k\to \l}}$ as a telescoping sum of terms of the form $\DKL{\nu D_{k \to i}\river \mu D_{k \to i}} - \DKL{\nu D_{k \to (i-1)}\river \mu D_{k \to (i-1)}},$  use the entropic independence of $\mu$ and its conditional distributions to derive inequalities involving these terms then chain these inequalities together to derive the desired entropy contraction. The main difference between our proof and the existing local-to-global framework \cite{CGM19,AL20,HL20,alimohammadi2021fractionally} is that we use $k\to 1$ entropy contraction whereas they use $2 \to 1$ contraction as the main building blocks. This difference is crucial because $2 \to 1$ contraction may not hold even for FLC distributions.

	\begin{proof}[Proof of \cref{thm:alpha frac LC MLSI}]
Consider an arbitrary distribution $\nu: \binom{[n]}{k} \to \R_{\geq 0}.$ We will show
\[ \DKL{\nu D_{k\to \l} \river \mu D_{k\to \l}} \leq \parens{1- \kappa} \DKL{\nu \river \mu} \]
	
	Consider the following process. We sample a set $S\sim \mu$ and uniformly at random permute its elements to obtain $X_1,\dots,X_k$. Let  $f(x) = x\log x.$ Define the random variable
	\[ \tau_i=f\parens*{\E*{\frac{\nu(S)}{\mu(S)}\given X_1,\dots,X_i}}=f\parens*{\frac{\sum_{S'\ni X_1,\dots, X_i}\nu(S')}{\sum_{S'\ni X_1,\dots,X_i}\mu(S')}}=f\parens*{\frac{\nu D_{k\to i}(\set{X_1,\dots,X_i})}{\mu D_{k\to i}(\set{X_1,\dots,X_i})}}. \]
	Note that $\tau_i$ is a ``function'' of $X_1,\dots,X_i$. It is not hard to see that 
	\[\DKL{\nu\river \mu}= \D_f{\nu\river \mu}=\E{\tau_k}-\E{\tau_0}=\sum_{i=0}^{k-1}(\E{\tau_{i+1}}-\E{\tau_{i}}). \]
	A convenient fact about this telescoping sum is that to obtain $\DKL{\nu D_{k\to \l}\river \mu D_{k\to\l}}$, one has to just sum over the first $\l$ terms instead of $k$:
	\[ \DKL{\nu D_{k\to \l}\river \mu D_{k\to \l}}=\E{\tau_\l}-\E{\tau_0}=\sum_{i=0}^{\l-1}(\E{\tau_{i+1}}-\E{\tau_{i}}). \]
	This is because the set $\set{X_1,\dots,X_\l}$ is distributed according to $\mu D_{k\to \l}$. So our goal of showing that $D_{k\to\l}$ has contraction boils down to showing that the last $k-\l$ terms in the telescoping sum are sufficiently large compared to the rest.
	
	Let $\Delta_i = (\E{\tau_{i+1}}-\E{\tau_{i}}) $ and
$\beta_i = \frac{1}{\alpha(k-i)-1}.$
Applying \cref{thm:entropic-independence} to $\mu$, we have
\begin{align*}
    \Delta_0 &\leq \beta_0(\Delta_1 + \dots + \Delta_{k-1})\\
    \intertext{Similarly, for $i=1, \dots, \l-1$, by applying \cref{thm:entropic-independence} to $\mu(\cdot \mid X_1, \dots, X_i)$ (recall that conditioning preserves $\alpha$-FLC), and averaging over $X_1, \dots, X_i \in [n]$, we get}
    \Delta_1 &\leq \beta_1(\Delta_2 + \dots + \Delta_{k-1})\\
    \dots & \dots\\
    \Delta_{\l-1} & \leq \beta_{\l-1}(\Delta_{\l}+\dots + \Delta_{k-1}).
\end{align*}
Then, it follows inductively that for all $0\leq i < \l \leq k - 1/\alpha$, 
\[\Delta_i \leq \beta_i \cdot  (\Delta_{\l} + \dots + \Delta_{k-1}) \cdot \prod_{j=i+1}^{\l-1}(\beta_j+1)\]
Hence
\begin{align*}
    \frac{\Delta_0 + \dots + \Delta_{k-1}}{\Delta_{\l} + \dots + \Delta_{k-1}}
    &= 1 + \frac{\Delta_0 + \dots + \Delta_{\l-1}}{\Delta_{\l}+\dots + \Delta_{k-1}}
    \leq 1+ \sum_{i=0}^{\l-1}\beta_i\cdot \prod_{j=i+1}^{\l-1}(\beta_j+1)\\
    &= \prod_{i=0}^{\l-1}(1+\beta_i) = \prod_{j=k-\l+1}^{k}\left(1 + \frac{1}{\alpha j -1}\right) = \prod_{j=k-\l+1}^{k} \frac{j}{j - 1/\alpha} \\
    &= \frac{\Gamma(k+1)/\Gamma(k+1-\l) }{\Gamma(k+1 -1/\alpha)/\Gamma(k+1-\l -1/\alpha) } \\
    &\leq \frac{(k+1)^{1/\alpha}}{  (k+1 - \l -  1/\alpha)^{1/\alpha -\lceil 1/\alpha \rceil } \prod_{i=0}^{\lceil 1/\alpha \rceil -1}(k-\l-i) } \\
\end{align*}
where $\Gamma(\cdot)$ denotes the Gamma function. The last inequality follows from \cite[Theorem~4.1(b)]{GammaBound}, which shows that $\frac{\Gamma(k+1) }{\Gamma(k+1 -1/\alpha)}\leq (k+1)^{1/\alpha}$ and  $\frac{\Gamma\parens{k+1 - \l -  1/\alpha }} {\Gamma\parens{k+1 - \l -  \lceil 1/\alpha \rceil}} \leq (k+1-\l - 1/\alpha)^{\lceil 1/\alpha \rceil - 1/\alpha}$, from which we obtain that  
\begin{align*}
    \Gamma(k+1-\l) &=\Gamma\parens{k+1 - \l - \lceil 1/\alpha \rceil}  \prod_{i=0}^{\lceil 1/\alpha \rceil -1}(k-\l-i)   \\
    &\geq \Gamma\parens{k+1 - \l -  1/\alpha }   (k+1 - \l -  1/\alpha )^{1/\alpha -\lceil 1/\alpha \rceil } \prod_{i=0}^{\lceil 1/\alpha \rceil -1}(k-\l-i).
\end{align*}

Note that for $1/\alpha \in \Z$, we can obtain the alternative and simpler looking bound of
\begin{align*}
    \frac{\Delta_0 + \dots + \Delta_{k-1}}{\Delta_{\l} + \dots + \Delta_{k-1}}
    &= \frac{\Gamma(k+1)/\Gamma(k+1-\l) }{\Gamma(k+1 -1/\alpha)/\Gamma(k+1-\l -1/\alpha) } \\
    &=\frac{k!/(k-\l)!}{(k-1/\alpha)!(k-\l-1/\alpha)!}\\
    &=\left.\binom{k}{1/\alpha}\middle/\binom{k-\l}{1/\alpha}\right. ,
\end{align*}
    as promised.

\end{proof}

\begin{proof}[Proof of \cref{thm:nonuniform FLC to entropy contraction}]
	   We use a similar strategy as in the proof of \cref{thm:entropic-independence}. Let $q_i : = \nu[X_i=1]/n$, and $q_{\bar{i}} : = \nu[X_i=-1]/n$. Note that $\nu^{\hom} D_{n\to 1} = q$. For convenience of notation, we will denote the multivariate generating polynomial $g_{\mu^{\hom}}$ simply by $g$. 
	   By \cref{lem:KL-dual}, for any $q \in \R^{2n}_{\geq 0}$ with $\sum_{i}q_i = 1$, 
	   \[ \inf\set*{\DKL{\nu^{\hom} \river \mu^{\hom}} \given \nu^{\hom} D_{n\to 1}=q} = -\log\parens*{\inf_{z_1,\dots, z_n, z_{\bar{1}}, \dots, z_{\bar{n}}} \frac{g(z_1,\dots,z_n, z_{\bar{1}}, \dots, z_{\bar{n}})}{z_1^{nq_1}\cdots z_n^{nq_n} z_{\bar{1}}^{n q_{\bar{1}}  }\cdots z_{\bar{n}}^{n q_{\bar{n}} } } }. \]
	   
	   Let $h(z_1, \dots, z_n, z_{\bar{1}}, \dots, z_{\bar{n}}) = g(z_1^{\alpha_1}, \cdots, z_n^{\alpha_n}, z_{\bar{1}}^{\alpha_1}, \cdots,  z_{\bar{n}}^{\alpha_n})^{1/\bar{\alpha}}$ with $\bar{\alpha} = \sum_{i=1}^n \alpha_i.$ Now, $h$ is 1-homogeneous and quasi-concave, thus concave, and we have the bound 
	   \begin{equation*}
	       \begin{split}
	           &\forall z_1, \dots, z_n,z_{\bar{1}}, \dots, z_{\bar{n}}>0: \\
	           &h(z_1, \dots, z_n, z_{\bar{1}}, \dots, z_{\bar{n}}) \leq h(1,\dots,1) + \sum_i \parens{\partial_i h(1,\dots, 1) (z_i-1) + \partial_{\bar{i}} h(1,\dots, 1) (z_{\bar{i}}-1)}.
	       \end{split}
	   \end{equation*}
	   
	   We next compute $\partial_i h(1,\dots, 1)$  and $ \partial_{\bar{i}} h(1,\dots, 1)$.  Let $p_i = \frac{\mu[X_i=1]}{n} , p_{\bar{i}} = \frac{\mu[X_i = -1]}{n}$ and note that 
	   \[\partial_i h(1,\dots, 1) = \alpha_i \parens*{\partial_i g_\mu(1,\dots,1)}\cdot\parens*{\frac{1}{\bar{\alpha}}\cdot g(1,\dots,1)^{1/\bar{\alpha}-1}}= \frac{n\alpha_i}{\bar{\alpha} } p_i, \]
	    \[\partial_{\bar{i}} h(1,\dots, 1) = \alpha_{\bar{i}} \parens*{\partial_{\bar{i}} g(1,\dots,1)}\cdot\parens*{\frac{1}{\bar{\alpha}}\cdot g(1,\dots,1)^{1/\bar{\alpha}-1}}=\frac{n\alpha_i}{\bar{\alpha} } p_{\bar{i}}. \]
	    Thus,
	    \[ \sum_{i=1}^{n} \parens{\partial_i h(1,\dots, 1) + \partial_{\bar{i}} h(1,\dots, 1)} = \sum_i \frac{\alpha_i}{\bar{\alpha} } n(p_i + p_{\bar{i}}) = \sum_i \frac{\alpha_i}{\bar{\alpha} } = 1 = h(1,\dots, 1),\]
	    and hence 
	    \[h(z_1, \dots, z_n, z_{\bar{1}}, \dots, z_{\bar{n}}) \leq \sum_i \frac{n\alpha_i}{\bar{\alpha} } \parens*{ p_i z_i + p_{\bar{i}} z_{\bar{i}}}.\]
	    Rewriting in terms of $g$, this gives 
	    \[g(z_1, \dots, z_n, z_{\bar{1}}, \dots, z_{\bar{n}})  \leq \parens*{ \sum_i \frac{n\alpha_i}{\bar{\alpha} } \parens*{p_i z_i^{1/\alpha_i}  + p_{\bar{i}} z_{\bar{i}}^{1/\alpha_i} } }^{\bar{\alpha}}.\]
	    
	    Therefore, 
		\[ \inf_{z_1,\dots, z_n, z_{\bar{1}}, \dots, z_{\bar{n}}>0} \frac{g(z_1,\dots,z_n, z_{\bar{1}}, \dots, z_{\bar{n}})}{z_1^{nq_1}\cdots z_n^{nq_n} z_{\bar{1}}^{n q_{\bar{1}}}\cdots z_{\bar{n}}^{n q_{\bar{n}}} }  \leq \inf_{z_1,\dots, z_n, z_{\bar{1}}, \dots, z_{\bar{n}}>0}\frac{\parens*{\sum_i \frac{n\alpha_i}{\bar{\alpha} } \parens*{ p_i z_i^{1/\alpha_i} + p_{\bar{i}} z_{\bar{i}}^{1/\alpha_i}} }^{\bar{\alpha} } }{z_1^{nq_1}\cdots z_n^{nq_n} z_{\bar{1}}^{n q_{\bar{1}}}\cdots z_{\bar{n}}^{n q_{\bar{n}}} } . \]
		Plugging in $z_i=(q_i/p_i)^{\alpha_i}$ and $z_{\bar{i}} = ({q_{\bar{i}}/{p_{\bar{i}}}})^{\alpha_i}$, we obtain
		\[  \inf_{z_1,\dots, z_n, z_{\bar{1}}, \dots, z_{\bar{n}}>0} \frac{g(z_1,\dots,z_n, z_{\bar{1}}, \dots, z_{\bar{n}})}{z_1^{nq_1}\cdots z_n^{nq_n} z_{\bar{1}}^{n q_{\bar{1}}}\cdots z_{\bar{n}}^{n q_{\bar{n}}} } \leq \prod_{i=1}^n (p_i/q_i)^{\alpha_i n q_i} (p_{\bar{i}}/q_{\bar{i}})^{\alpha_i n q_{\bar{i}}} . \]
		Taking $\log$ and negating gives 
		\begin{align*} 
		\DKL{\nu \river \mu} &= \DKL{\nu^{\hom} \river \mu^{\hom}}\geq -\log\left( \prod_{i=1}^n (p_i/q_i)^{\alpha_i n q_i} (p_{\bar{i}}/q_{\bar{i}})^{\alpha_i n q_{\bar{i}}}\right)\\
		&= \sum_{i=1}^{n} \alpha_i \parens{(nq_i) \log(q_i/p_i) + (nq_{\bar{i}}) \log (q_{\bar{i}}/ p_{\bar{i}})  } = \sum_{i=1}^{n}\alpha_i \DKL{\nu_i \river \mu_i}. 
		\end{align*}
		This is exactly the claimed bound in \cref{thm:nonuniform FLC to entropy contraction}.
	\end{proof}

	
	By a similar argument, we obtain the following generalization of \cref{thm:nonuniform FLC to entropy contraction} to any distribution $\mu$ on $Q^n$, for arbitrary finite set $Q.$

\begin{proposition} \label{prop:multispin nonuniform FLC to entropy contraction}
		Let $Q$ be a finite set.
	    Let $\mu$ be a distribution on $Q^n.$ We can view $\mu$ as the homogeneous distribution on  $\binom{\Omega}{n}$ with $\Omega$ be the disjoint union of $Q_1, \dots, Q_n,$ where each $Q_i$ is a copy of $Q.$ Suppose that there exist $\alpha_1,\dots, \alpha_n \in (0,1]$ for which $\mu$ is $(\underbrace{\alpha_1, \dots, \alpha_n,\dots, \alpha_1, \dots, \alpha_n}_{\abs{Q} \text{ times} })$-FLC. 
	    
	    Let $\nu$ be a distribution on $Q^n$ and for $i \in [n]$, let $\mu_i$ (respectively $\nu_i$) denote the marginal distribution of the $i^{th}$ coordinate under $\mu$ (respectively $\nu$). Then,
	    \[\sum_{i=1}^{n}\alpha_{i}\DKL{\nu_i \river \mu_i} \leq \DKL{\nu \river \mu}.\]
	\end{proposition}
	\begin{proof}
	Let $ q_{i,\epsilon_i} = \nu[X_i = \epsilon_i]/n, p_{i,\epsilon_i} = \mu[X_i = \epsilon_i]/n.$ Note that $\nu D_{n \to 1 } = q.$ By the same argument as in proof of \cref{thm:nonuniform FLC to entropy contraction}
	\begin{align*}
	    \inf\set*{\DKL{\nu \river \mu} \given \nu D_{n\to 1}=q} &= -\log\parens*{\inf_{z_{i,\epsilon_i}: i\in [n], \epsilon_i \in Q} \frac{g(z_{i,\epsilon_i}) }{\prod_{(i,\epsilon_i) \in [n] \times Q}z_{i,\epsilon_i}^{nq_{i,\epsilon_i}} } }\\
	    &\geq  -\log\parens*{ \inf_{z_{i,\epsilon_i}: i\in [n], \epsilon_i \in Q} \frac{\parens*{ \sum_i \frac{n \alpha_i}{\bar{\alpha}} \parens*{\sum_{\epsilon_i \in Q} p_{i, \epsilon_i} z_{i,\epsilon_i}^{1/\alpha_i} }}^{\bar{\alpha}}  }{\prod_{(i,\epsilon_i) \in [n] \times Q}z_{i,\epsilon_i}^{nq_{i,\epsilon_i}} }   }\\
	    &\geq - \log \parens*{\prod_{i=1}^n \prod_{\epsilon_i \in Q} (p_{i,\epsilon_i}/q_{i,\epsilon_i})^{\alpha_i n q_{i, \epsilon_i}} } \\
	    &=\sum_{i=1}^n \alpha_i \sum_{\epsilon_i \in Q} (nq_{i,\epsilon_i}) \log \frac{q_{i,\epsilon_i}}{p_{i,\epsilon_i}}\\
	    &= \sum_{i=1}^n \alpha_i \DKL{\nu_i \river \mu_i} 
	\end{align*}
	\end{proof}

\end{document}